\documentclass[12pt]{article}

\usepackage{amsmath}
\usepackage{amsthm}
\usepackage{amsfonts}
\usepackage{amssymb}
\usepackage{fullpage}

\usepackage{authblk}

\usepackage{ifthen}
\usepackage{xspace}

\usepackage{microtype}
\usepackage{fancyhdr}
\usepackage{varioref}
\usepackage{tgtermes}
\usepackage[scale=.92]{tgheros}
\usepackage{tgcursor}

\usepackage[algo2e,tworuled,vlined]{algorithm2e}
\DontPrintSemicolon
\LinesNumbered
\usepackage{algorithm}

\usepackage{bussproofs}
\usepackage{bm}
\usepackage{amsthm}
\usepackage{mathrsfs}

\newboolean{conferenceversion}
\setboolean{conferenceversion}{false}

\newcommand{\introduceterm}[1]{{\emph{#1}}}

\newtheorem{theorem}{Theorem}[section]

\newtheorem{proposition}[theorem]{Proposition}
\newtheorem{lemma}[theorem]{Lemma}

\newtheorem{fact}[theorem]{Fact}
\newtheorem{claim}[theorem]{Claim}

\newtheorem{definition}[theorem]{Definition}
\theoremstyle{remark}
\newtheorem{remark}{Remark}

\newcommand{\refsec}[1]{Section~\ref{#1}}
\newcommand{\refth}[1]{Theorem~\ref{#1}}
\newcommand{\refthm}[1]{Theorem~\ref{#1}}
\newcommand{\refdef}[1]{Definition~\ref{#1}}
\newcommand{\refeq}[1]{\eqref{#1}}
\numberwithin{equation}{section}

\newcommand{\eqperiod}{\enspace .}
\newcommand{\eqcomma}{\enspace ,}
\newcommand{\formuladots}{\cdots}

\newcommand{\ie}{i.e.,\ }

\newcommand{\aas}{asymptotically almost surely\xspace}

\newcommand{\ErdosRenyi}{Erd\H{o}s-R\'enyi\xspace}

\newcommand{\E}{\mathrm{e}}%
\newcommand{\Rplus}     {\mathbb{R}^{+}}
\newcommand{\Nplus}     {\mathbb{N}^{+}}

\newcommand{\ceiling}[1]{\lceil #1 \rceil}

\newcommand{\set}[1]{\{ #1 \}}
\newcommand{\Set}[1]{\bigl\{ #1 \bigr\}}
\newcommand{\setdescr}[3][\mid]{\set{ #2 #1 #3 }}
\newcommand{\Setsize}[1]{\bigl\lvert#1\bigr\rvert}
\newcommand{\setsize}[1]{\lvert#1\rvert}
\newcommand{\Union}{\bigcup}

\newcommand{\bigoh}[1]{\mathrm{O} ( #1 )}
\newcommand{\littleoh}[1]{\mathrm{o} ( #1 )}
\newcommand{\Bigomega}[1]{\Omega \bigl( #1 \bigr)}
\newcommand{\bigomega}[1]{\Omega ( #1 )}

\newcommand{\complclassformat}[1]%
        {\textrm{\upshape{\textsf{#1}}}\xspace}
\newcommand{\Pclass}{\complclassformat{P}}
\newcommand{\NP}{\complclassformat{NP}}
\newcommand{\Wclass}[1]{\complclassformat{W}\kern-.01em$[#1]$\xspace}
\newcommand{\FPTclass}{\complclassformat{FPT}}

\newcommand{\prob}[2][]{\Pr_{#1} [ #2 ]}
\newcommand{\Prob}[2][]{\Pr_{#1} \bigl[ #2 \bigr]}
\newcommand{\PROB}[2][]{\Pr_{#1} \left[ #2 \right]}

\newcommand{\twincommandJN}[6]%
    {#1#2#3\vphantom{#2#5}\mspace{-2.05mu}#4.#5#6}
\newcommand{\Condprob}[3][]%
    {\Pr_{#1}\twincommandJN{\bigl[}{#2}{\bigl|}{\bigr}{\,#3}{\bigr]}}
\newcommand{\CONDPROB}[3][]%
    {\Pr_{#1}\twincommandJN{\left[}{#2}{\left|}{\right}{\,#3}{\right]}}
\newcommand{\expect}[1]{\mathbb{E}[#1]} %
\newcommand{\Expect}[1]{\mathbb{E}\big[#1\big]} %
\newcommand{\sample}[2]{#1\sim#2}
\newcommand{\event}[1]{\mathsf{#1}} %
\newcommand{\eventE}{\event{E}}
\newcommand{\eventA}{\event{P}}
\newcommand{\randomvar}[1]{\boldsymbol{#1}} %
\newcommand{\randomvarG}{\randomvar{G}}
\newcommand{\Xrand}{\randomvar{X}}%
\newcommand{\randomvaralpha}{\randomvar{\alpha}}
\newcommand{\distribution}[1]{\mathscr{#1}} %
\newcommand{\distributionD}{\distribution{D}}

\newcommand{\proofstd}{\pi}
\newcommand{\emptycl}{\bot}
\newcommand{\formf}{\ensuremath{F}}
\newcommand{\varx}{\ensuremath{x}}

\newcommand{\clb}{\ensuremath{B}}
\newcommand{\clc}{\ensuremath{C}}
\newcommand{\cld}{\ensuremath{D}}

\newcommand{\prooflength}{L}
\newcommand{\lengthsize}{length\xspace}
\newcommand{\bpbit}{\gamma}
\newcommand{\bprogpi}{P}

\newcommand{\vars}[1]{\mathit{Vars}({#1})}
\newcommand{\restrict}[2]{{{#1}\!\!\upharpoonright_{#2}}}

\newcommand{\cliqueformulaname}{\textit{Clique}}%
\newcommand{\cliquevariant}[1]{_{\mathrm{#1}}} %
\newcommand{\clique}[2]{\cliqueformulaname(#1,#2)}
\newcommand{\wclique}[2]{\cliqueformulaname^*(#1,#2)}
\newcommand{\cliquecount}[2]{\cliqueformulaname\cliquevariant{count}(#1,#2)}
\newcommand{\cliquemap}[2]{\cliqueformulaname\cliquevariant{map}(#1,#2)}
\newcommand{\cliqueblock}[2]{\cliqueformulaname\cliquevariant{block}(#1,#2)}
\newcommand{\cliquemembervar}[1]{x_{#1}}
\newcommand{\countingvar}[2]{y_{#1,#2}}
\newcommand{\icliquemembervar}[2]{x_{#1,#2}}

\newcommand{\cliquesize}{k} %
\newcommand{\cliqueaxiom}{clique axiom\xspace}
\newcommand{\cliqueaxioms}{clique axioms\xspace}
\newcommand{\edgeaxiom}{edge axiom\xspace}
\newcommand{\edgeaxioms}{edge axioms\xspace}
\newcommand{\fun}{functionality\xspace}
\newcommand{\funaxiom}{functionality axiom\xspace}
\newcommand{\funaxioms}{functionality axioms\xspace}

\newcommand{\orderaxioms}{ordering axioms\xspace}

\newcommand{\Gp}{G} %
\newcommand{\Vp}{V} %
\newcommand{\Ep}{E} %
\newcommand{\np}{n} %
\newcommand{\randG}[2]{\distribution{\Gp}(#1,#2)} %
\newcommand{\neigh}[1]{N(#1)} %
\newcommand{\neighc}[2]{\widehat{N}_{#2}(#1)}

\newcommand{\Neighcsize}[2]{\Setsize{\neighc{#1}{#2}}}
\newcommand{\commonneighbourhood}[2]{\neighc{#1}{#2}} %
\newcommand{\neighcV}[1]{\widehat{N}(#1)}

\newcommand{\setzeros}[2]{V^0_{#1}(#2)}
\newcommand{\setones}[2]{V^1_{#1}(#2)}
\newcommand{\unionsetones}{V^1(a)}
\newcommand{\betaone}{\beta^1}

\newcommand{\genericforR}[1]{$#1$\nobreakdash-neigh\-bour-dense\xspace}
\newcommand{\generic}[2]{$(#1,#2)$\nobreakdash-neigh\-bour-dense\xspace}
\newcommand{\genericNOparam}{neighbour-dense\xspace}
\newcommand{\genericnessNOparam}{neighbour-denseness\xspace}
\newcommand{\GenericNOparam}{Neighbour-dense\xspace}
\newcommand{\robgeneric}[4]{$(#1,#2,#3,#4)$\nobreakdash-mostly neighbour-dense\xspace}

\newcommand{\RobgenericNOparam}{Mostly neighbour-dense\xspace}
\newcommand{\cliquedense}[5]{$(#1,#2,#4,#5)$\nobreakdash-clique-dense\xspace}
\newcommand{\cliquedenseNOparam}{clique-dense\xspace}
\newcommand{\CliquedenseNOparam}{Clique-dense\xspace}
\newcommand{\cliquedensenessNOparam}{clique-denseness\xspace}

\newcommand{\rp}{r}
\newcommand{\qp}{q}
\newcommand{\spar}{s}
\newcommand{\Wp}{W}
\newcommand{\Rp}{R}
\newcommand{\Sp}{S}

\newcommand{\newepsilon}{\varepsilon}
\newcommand{\oldxi}{\newepsilon}
\newcommand{\largerp}{r'}
\newcommand{\oldeta}{\xi}

\newcommand{\useful}{useful\xspace}
\newcommand{\frugally}{frugally\xspace}
\renewcommand{\frugally}{usefully\xspace}

\newcommand{\alphar}{\randomvaralpha} %
\newcommand{\Ddist}{\distributionD} %

\newcommand{\tp}{t} %

\newcommand{\randomvarR}{R(\alphar)} %
\newcommand{\setWR}{\mathcal{W}} %

\newcommand{\disj}[1]{$#1$-dis\-joint\xspace} %
\newcommand{\neighcineq}[1]{\event{F}(#1)}
\newcommand{\neighcineqp}[1]{\event{F}'(#1)}
\newcommand{\yp}{\ell}

\begin{document}

\title{Clique Is Hard on Average for Regular Resolution}

\author[1]{Albert Atserias}
\author[1]{Ilario Bonacina}
\author[2]{Susanna F. de Rezende}
\author[3]{Massimo Lauria}
\author[4]{Jakob Nordstr\"om}
\author[5]{Alexander Razborov}
\affil[1]{Universitat Polit\`ecnica de Catalunya}
\affil[2]{Institute of Mathematics of the Czech Academy of Sciences}
\affil[3]{Sapienza - Universit\`a di Roma}
\affil[4]{University of Copenhagen and Lund University}
\affil[5]{University of Chicago and Steklov Mathematical Institute}

\renewcommand\Authands{ and }


\maketitle

\begin{abstract}
  We prove that for $k \ll \sqrt[4]{n}$ regular resolution requires length $n^{\bigomega{k}}$ to establish that
  an Erd\H{o}s--Rényi graph with appropriately chosen edge density
  does not contain a $k$-clique. This lower bound is optimal up to the
  multiplicative constant in the exponent, and  also implies
  unconditional $n^{\bigomega{k}}$ lower bounds on running time for
  several state-of-the-art algorithms for finding maximum cliques in
  graphs.
\end{abstract}

\section{Introduction}
\label{sec:intro}

Deciding whether a graph has a $k$-clique is one of the most basic
computational problems on graphs, and has been extensively studied in
computational complexity theory ever since it appeared in Karp's list
of $21$ \NP-complete problems~\cite{Karp72Reducibility}.  Not only is
this problem widely believed to be intractable to solve
exactly (unless $\Pclass = \NP$), there does not even exist any
polynomial-time algorithm for approximating the maximum size of a
clique to within a factor~$n^{1-\epsilon}$ for any constant
$\epsilon > 0$, where $n$ is the number of vertices in the graph
\cite{Hastad99Clique,Zuckerman07LinearDegreeExtractors}.  Furthermore,
the problem appears to be hard not only in the worst case but
also on average in the \ErdosRenyi random graph model---we know of no
efficient algorithms for finding cliques of maximum size
asymptotically almost surely on random graphs with appropriate edge
densities~\cite{Karp76ProbabilisticAnalysis,Rossman10Thesis}.

In terms of upper bounds, the $k$-clique problem can clearly be solved
in time roughly~$n^k$ simply by checking if any of the $\binom{n}{k}$
many sets of vertices of size~$k$ forms a clique. This takes polynomial time
if $k$~is constant. This can be improved slightly
to~$\bigoh{n^{\omega k/3}}$, where $\omega\leq 2.373$ is the
matrix multiplication exponent, using algebraic
techniques~\cite{NP85CplxSubgraph}, although in practice such algebraic
algorithms are outperformed by combinatorial
ones~\cite{Vassilevska09EfficientClique}.

The motivating problem behind this work is to determine the exact time
complexity of the clique problem when $k$~is given as a parameter.  As
noted above, all known algorithms require time $n^{\Omega(k)}$.  It appears
quite likely that some dependence on~$k$ is needed in the exponent,
since otherwise we have the parameterized complexity collapse
$\FPTclass = \Wclass{1}$ \cite{DF95FPTandCompletenessII}.  Even more
can be said if we are willing to believe the Exponential Time
Hypothesis (ETH) \cite{IP01Complexity}---then the exponent has to
depend linearly on~$k$~\cite{CHKX04LinearFPTreductions}, so that the
trivial upper bound is essentially tight.

Obtaining such a lower bound unconditionally would, in particular,
imply $\Pclass \neq \NP$, and so
currently seems completely out of reach.
But is it possible to prove $n^{\Omega(k)}$ lower bounds in restricted but
nontrivial models of computation?
For circuit complexity, this challenge has been met for circuits that
are of bounded depth~\cite{Rossman08ConstantDepth} or
are monotone~\cite{Rossman14Monotone}.
In this paper we focus on computational models that are powerful
enough to capture several algorithms that are used in practice.

When analysing such algorithms, it is convenient to view the execution
trace as a proof establishing the maximum clique size for the input
graph.  In particular, if this graph does not have a $k$-clique, then
the trace provides an efficiently verifiable proof of the statement
that the graph is \mbox{$k$-clique}-free.  If one can establish a lower
bound on the \lengthsize of such proofs, then this implies a lower bound on
the running time of the algorithm, and this lower bound holds even if
the algorithm is a non-deterministic heuristic that somehow magically
gets to make all the right choices.
This brings us to the topic of
\emph{proof complexity}~\cite{CR79Relative},
which can be viewed as the study of upper and lower bounds in
restricted nondeterministic computational models.

Using a standard reduction from \mbox{$k$-clique} to SAT,
we can translate the problem of \mbox{$k$-cliques} in
graphs to that of satisfiability of formulas in conjunctive
normal form (CNF).
If an algorithm for finding $k$-cliques is run on a graph $G$ that
is \mbox{$k$-clique}-free, then we can extract a proof of the
unsatisfiability of the corresponding CNF formula---the $k$-clique
formula on~$G$---from the execution trace of the algorithm.
Is it possible to show
any non-trivial lower bound on the \lengthsize of such proofs?
Specifically,
does
the \emph{resolution} proof system---the method
of reasoning underlying state-of-the-art SAT
solvers~\cite{BS97UsingCSP,MS99Grasp,MMZZM01Engineering}---require
\lengthsize $n^{\Omega(k)}$, or at least $n^{\omega_k(1)}$ (i.e. the exponent as a function of $k$ is not bounded by a constant), to prove the
absence of $k$-cliques in a graph?  This question was asked
in, e.g.,~\cite{BGLR12Parameterized} and remains open.

The hardness of $k$-clique formulas for resolution is also a problem
of intrinsic interest in proof complexity, since these formulas escape
known methods of proving resolution lower bounds for
a range of interesting values of~$k$ including $k = \bigoh{1}$.
In particular,
the interpolation
technique~\cite{Krajicek97Interpolation,Pudlak97LowerBounds}, the
random restriction method~\cite{BP96Simplified}, and the size-width
lower bound~\cite{BW01ShortProofs}
all seem to fail.

To make this more precise, we should mention that some previous works
do use the size-width method, but only for very large~$k$.  It was shown
in~\cite{BIS07IndependentSets} that for $n^{5/6} \ll k \leq n/3$
 resolution requires \lengthsize
$\exp\mathopen{}\big({n^{\Omega(1)}}\big)\mathclose{}$ to certify that a dense enough \ErdosRenyi
random graph is $k$-clique-free.
The constant hidden in the $\Omega(1)$ increases with the density of the
graph and, in particular, for very dense graphs and $k=n/3$
the \lengthsize required is $2^{\Omega(n)}$.
Also, for a specially tailored CNF encoding, where the  $i$th~member of the
claimed $k$-clique is encoded in binary by $\log n$ variables, a lower bound of
$n^{\Omega(k)}$ for $k \leq \log n$ can be extracted from a careful reading
of~\cite{LPRT17ComplexityRamsey}.
However, in the more natural unary encodings, where indicator variables
specify whether a vertex is in the clique,
the size-width method cannot yield more than a~$2^{\Omega(k^2/n)}$
lower bound since there are resolution proofs of width~$\bigoh{k}$. This bound becomes trivial when $k \leq \sqrt{n}$.

In the restricted subsystem of \emph{tree-like resolution}, optimal
$n^{\bigomega{k}}$ \lengthsize lower bounds were established
in~\cite{BGL13ParameterizedDPLL} for $k$-clique formulas on complete
\mbox{$(k-1)$-partite} as well as on average for \ErdosRenyi random
graphs of appropriate edge density.
There is no hope to get hard instances for general resolution from complete
\mbox{$(k-1)$-partite} graphs, however---in the same paper it was shown
that all instances from the more general class of
$(k-1)$-colourable graphs are easy for resolution.
A closer study of these resolution proofs reveals that they are
\emph{regular}, meaning that if the proof is viewed
as a directed acyclic graph (DAG), then no variable is eliminated more
than once on any source-to-sink path.

More generally, regular resolution is an interesting and non-trivial model to analyse
for the $k$-clique problem since it captures
the reasoning used in many state-of-the-art algorithms
used in practice
(for a survey,
see, e.g., \cite{Prosser12Exact,McCreesh17Thesis}). Nonetheless, it has
remained consistent with state-of-the-art knowledge that
for $k\leq n^{5/6}$
regular resolution might be able to certify \mbox{$k$-clique}-freeness
in polynomial \lengthsize independent of the value of~$k$.

\paragraph{Our contributions}
We prove optimal
$n^{\bigomega{k}}$
average-case  lower bounds for regular resolution proofs of
unsatisfiability for $k$-clique formulas on \ErdosRenyi random graphs.

\begin{theorem}
  [Informal]
	For any integer $k\ll \sqrt[4]{n}$, given an $n$-vertex graph~$G$
        sampled at
	random from the \ErdosRenyi model with the appropriate edge density,
	regular resolution asymptotically almost surely requires
        \lengthsize $n^{\Omega({k})}$
	to certify that $G$ does not contain a $k$-clique.
	\label{thm:informal}
\end{theorem}

\ifthenelse{\boolean{conferenceversion}}
{In order to make this formal,
we need to define how the problem is encoded:
depending on the formula considered, the exact statement of what
we can prove differs.
In this conference paper we consider the simpler encoding for which we can
prove an $n^{\Omega({k})}$ lower bound for $k\ll \sqrt{n}$.
For a stronger encoding, which in particular captures this simpler one, we
prove the above result in the full-length version of this paper.}
{}

At a high level, the proof is based on a bottleneck counting
argument in the style of~\cite{Haken85Intractability} with a slight twist
that was introduced in \cite{RWY02Read-once}.
In its classical form, such a proof takes four steps. First,
one defines
a distribution of random source-to-sink paths on the DAG
representation of the
proof. Second, a subset of the vertices of the DAG is identified%
---the set of \emph{bottleneck nodes}---such that any random path must
necessarily pass through at least one such node.
Third, for any fixed bottleneck
node, one shows that it is very unlikely that a random path passes
through
this particular node.
Given this, a final union bound argument yields the conclusion
that the DAG must have many bottleneck nodes, and so the resolution
proof must be long.

The twist in our argument is that, instead of single bottleneck nodes,
we need to define \emph{bottleneck pairs} of nodes.
We then argue that
any random path passes through at least one such pair but that few random
paths pass through any fixed pair; the latter part is based on Markov chain-type
reasoning similar to \cite[Theorems 3.2, 3.5]{RWY02Read-once}. Furthermore,
it crucially
relies on the graph satisfying a certain combinatorial property,
which captures the idea that the common neighbourhood of a
small set of vertices is well distributed across the graph.
Identifying this combinatorial property is a key contribution of our work.
In a separate argument (that, surprisingly, turned out to be much more elaborate
than most arguments of this kind) we then establish that
\ErdosRenyi random graphs  of the appropriate edge density satisfy
this property asymptotically almost surely.
Combining these two facts yields our average-case lower bound.

The idea of counting bottlenecks of more than one node comes from~\cite{RWY02Read-once}
and was also used in~\cite{BBI16TimeSpace}.

Another contribution of this paper is a relatively simple observation that not only is regular
resolution powerful enough to distinguish graphs that contain
$k$-cliques from $(k-1)$-colourable
graphs~\cite{BGL13ParameterizedDPLL}, but it can also distinguish them
from graphs that have a homomorphism to any fixed
graph~$H$ with no $\cliquesize$-cliques.
\paragraph{Recent Developments}
A preliminary version of this work appeared in the proceedings of the
STOC'18 conference~\cite{ABdRLNR18Clique}.  The techniques used there
to prove the $n^{\Omega(k)}$ average-case lower bound for regular
resolution were recently extended %
by Pang~\cite{Shuo19LARGE} to work for a proof system %
between regular and general resolution.  In the same paper, Pang also
shows a $2^{\Omega(k^{(1-\epsilon)})}$ %
resolution lower bound %
for $k$-clique formulas on \ErdosRenyi\ random graphs, for $k=n^{c}$,
$c<1/3$ and $\epsilon>0$.

Regarding the proof complexity of $k$-clique formulas for tree-like resolution, 
the lower bounds from~\cite{BGL13ParameterizedDPLL} and~\cite{LPRT17ComplexityRamsey} 
were simplified and unified in \cite{Lauria18Clique}. 
The resolution lower bound in~\cite{LPRT17ComplexityRamsey}
for $k$-clique formulas on \ErdosRenyi random graphs 
under the binary encoding was recently extended 
to an $n^{\bigomega{k}/d(s)}$ lower bound for $\mathrm{Res}(s)$, where $s = o((\log\log n)^{1/3})$
and $d(s)$ is a doubly exponential function~\cite{DGGM20ProofComplexity}.
\paragraph{Paper outline}
The rest of this paper is organized as follows.
Section~\ref{sec:preliminaries} presents some preliminaries.
We show that some nontrivial $\cliquesize$-clique instances are 
easy for regular resolution in
Section~\ref{sec:shortproofs}.
Section \ref{sec:recap-lower-bounds} contains the formal statement of 
the lower bounds we prove for \ErdosRenyi random graphs.
In Section~\ref{sec:lowerbound1} we define a combinatorial property of
graphs and show that clique formulas on such graphs are hard for
regular resolution,    
and the proof that \ErdosRenyi random graphs satisfy this property
asymptotically almost surely is in Section~\ref{sec:lowerbound2}.
\ifthenelse{\boolean{conferenceversion}}{}{%
Section~\ref{sec:algorithms} explains why our
results imply lower bounds on the running time of state-of-the-art algorithms
for $\cliquesize$-clique.
}
We conclude in Section~\ref{sec:open-problems} with a discussion of
open problems.

\section{Preliminaries}
\label{sec:preliminaries}
We write  $\Gp=(\Vp,\Ep)$ to denote a graph with vertices~$\Vp$ and
edges~$\Ep$, where $\Gp$ is always undirected, without loops and
multiple edges.
Given a vertex $v \in \Vp$,
we write
$\neigh{v}$
to denote the set of \emph{neighbours of $v$}.
For a set of vertices $\Rp \subseteq \Vp$ we write
$\neighcV{\Rp}=\bigcap_{v\in \Rp}\neigh{v}$ to denote the set of
\emph{common neighbours of $\Rp$}.
For two sets of vertices $\Rp \subseteq \Vp$ and $\Wp \subseteq \Vp$
we write
$\neighc{\Rp}{\Wp}=\neighcV{\Rp}\cap \Wp$
to denote the set of \emph{common neighbours of $\Rp$ inside $\Wp$}.
For a set $U \subseteq \Vp$ we denote by $\Gp[U]$ the subgraph of $G$ induced
by the set $U$.
For $n \in \Nplus$ we write
$[n] = \set{1,\ldots,n}$.
We say that $\Vp_1 \stackrel .\cup \Vp_2 \stackrel .\cup \cdots \stackrel .\cup \Vp_k = \Vp$ is a \emph{balanced}
$k$-partition of~$\Vp$
if for all $i,j\in [k]$ it holds that $\setsize{\Vp_i}\leq \setsize{\Vp_j} + 1$.
All logarithms are natural (base $\E$) if not specified otherwise.

\paragraph{Probability and \ErdosRenyi random graphs}
We often denote random variables in boldface and write
$\sample{\Xrand}{\distributionD}$
to denote that
$\Xrand$ is
sampled from the distribution $\distributionD$.
A {\em $p$-biased coin}, or a {\em Bernoulli variable}, is the outcome
of a
coin flip that yields $1$ with probability $p$ and
$0$ with probability $1-p$. We use the
special case of Markov's
inequality
saying that if $\Xrand$ is
non-negative,
then $\prob{\Xrand
  \geq 1} \leq \expect{\Xrand}$. We also need the following special case of the
multiplicative Chernoff bound:
if $\Xrand$ is a binomial random variable (i.e., the sum of
i.i.d.\ Bernoulli variables)
with expectation $\mu = \expect{\Xrand}$, then $\prob{\Xrand \leq \mu/2} \leq
\E^{-\mu/8}$.

We consider the \ErdosRenyi distribution $\randG{\np}{p}$ of random graphs
on a fixed set $V$ of $n$ vertices.  A random graph sampled from
$\randG{\np}{p}$ is produced by placing each potential edge $\{u,v\}$
independently
with probability $p$, $0 \leq p \leq 1$ (the edge probability~$p$
may be a function of $n$).
A property of graphs is said to hold
\emph{\aas} on $\randG{\np}{p(n)}$ if it
holds with probability that approaches $1$ as $n$ approaches infinity.

For a positive integer $k$, let $\Xrand_k$ be the random variable that
counts the number of $k$-cliques in
a random graph
from ${\randG{\np}{p}}$.  It follows from Markov's
inequality that
\aas
there are no $k$-cliques in $\randG{\np}{p}$
whenever $p$ and $k$ are such that
$\expect{\Xrand_k}=p^{\binom{k}{2}}\binom{\np}{k}$ approaches $0$ as $n$
approaches infinity.  This is the case, for example, if
$p = \np^{-2\eta/(k-1)}$ for $k \geq 2$ and $\eta > 1$. 
Actually, the clique number, i.e. the size of the largest clique, $\omega(G)$ for a graph $G$ sampled from $\mathscr{G}(n,p)$ is a well studied quantity and very strong concentrations bounds are known for it. For instance, one of the first concentration results is that $\omega(G)=(2-o(1))\log_{\frac{1}{p}}(n)$ with probability $1$ as $n\rightarrow \infty$ (see for instance \cite{BE.76}). 

\paragraph{CNF formulas and resolution}
A \introduceterm{literal} over a Boolean variable $\varx$ is either
the variable $\varx$ itself (a \introduceterm{positive literal}) or
its negation $\lnot \varx$ (a
\introduceterm{negative literal}).
A \introduceterm{clause}
$\clc = \ell_1 \lor \formuladots \lor \ell_{w}$ is a disjunction
of literals; we say that the
\introduceterm{width} of $\clc$ is~$w$. The empty clause will be denoted
by $\emptycl$.
A \introduceterm{CNF formula}
$F = \clc_1 \land \formuladots \land \clc_m$ is a conjunction
of clauses.
We think of clauses as sets of literals and of CNF formulas as sets of
clauses, so that order is irrelevant and there are no repetitions.
For a formula $F$ we denote by $\vars{F}$ the set of variables of~$F$.

A \introduceterm{resolution derivation} from
a CNF formula~$F$
is as an ordered sequence of clauses
$\proofstd = (\cld_1, \dotsc, \cld_{L})$
such that
for each $i \in [L]$
either $\cld_i$
is a clause in $\formf$
or there exist $j < i$ and $k < i$ such that
$\cld_i$ is derived from $\cld_j$ and $\cld_k$
by the
\introduceterm{resolution rule}
\begin{equation}
\label{eq:resolution-rule}
\AxiomC{$\clb \lor x$}
\AxiomC{$\clc \lor \lnot{x}$}
\BinaryInfC{$\clb \lor \clc$}
\DisplayProof
\eqcomma
\end{equation}
$\cld_i = \clb\lor\clc,\ \cld_j = \clb \lor x,\ \cld_k= \clc \lor \lnot{x}$.
We refer to $\clb \lor \clc$ as the \introduceterm{resolvent} of $\clb
\lor x$ and~$\clc \lor \lnot{x}$ over $x$, and to $x$ as the
\emph{resolved variable}. The \introduceterm{\lengthsize}  (or
\introduceterm{size}) of a resolution derivation $\proofstd = (\cld_1, \dotsc, \cld_{L})$ is~$L$ and it is denoted by $\setsize{\proofstd}$.
A \introduceterm{resolution  refutation} of~$F$, or
\introduceterm{resolution proof} for (the unsatisfiability of)~$F$,
is a resolution derivation from~$F$ that
ends in the empty clause~$\emptycl$.

A resolution derivation $\proofstd = (\cld_1, \dotsc, \cld_{L})$ can
also be viewed as a labelled
DAG
with the set of nodes
$\{1,\ldots,L\}$ and
edges
$(j,i)$,
$(k,i)$
for each application of the resolution rule deriving
$\cld_i$ from $\cld_j$ and~$\cld_k$.
Each node $i$ in this DAG is labelled by its associated clause
$\cld_i$, and each non-source node is also labelled by the resolved
variable in its associated derivation step in the refutation.
A resolution refutation is called \introduceterm{regular} if along any
source-to-sink path in its associated DAG every variable is resolved
at most once.

For a partial assignment $\rho$ we say that a clause $C$
 \introduceterm{restricted by~$\rho$}, denoted $\restrict{C}{\rho}$, is
 the trivial $1$-clause if any of the literals in~$C$ is satisfied by~$\rho$
 or otherwise is $C$ with all falsified literals removed. We extend this definition
 to CNFs in the obvious way: $\restrict{(C_1\land\dots\land C_m)}{\rho} = \restrict{C_1}{\rho}\land\dots\land
\restrict{C_m}{\rho}$.
Applying a restriction
preserves (regular) resolution derivations. To see this, observe
that in every application of the
resolution rule, the restricted consequence either becomes identically 1, or it is obtained, as before, by
resolving the two restricted premises, or it is a weakening of one of them, but weakenings can be removed at no cost.
Thus, we have:
\begin{fact}
  \label{fact:restrict}
  Let $\proofstd$ be a (regular) resolution refutation of a CNF
  formula $F$. For any partial assignment $\rho$ to the variables of
  $F$ there is an efficiently constructible (regular) resolution
  refutation $\restrict{\proofstd}{\rho}$ of the CNF formula
  $\restrict{F}{\rho}$, so that the \lengthsize of
  $\restrict{\proofstd}{\rho}$ is at most the \lengthsize of $\proofstd$.
\end{fact}

\paragraph{Branching programs}
A branching program on variables $x_1, \ldots, x_n$ is a
DAG
that has
one source node
and where
every non-sink node is
labelled by one of the variables $x_1, \ldots, x_n$ and has exactly two
outgoing edges labelled $0$ and $1$. %
The size of a branching program is the total number of nodes in the
graph.
In a \emph{read-once branching program} it holds in addition that
along every path every variable appears as a node label at most once.

For each node $a$ in a branching program, let $X(a)$ denote the
variable that labels $a$, and
let $a^0$ and $a^1$ be the nodes that are reached from $a$ through the
edges labelled~$0$ and $1$, respectively.
A truth-value assignment $\sigma : \{x_1,\ldots,x_n\} \rightarrow
\{0,1\}$ determines a path in a branching program in the following
way. The path starts at the source node. At an internal node $a$, the
path is extended along the edge labelled~$\sigma(X(a))$ so that the next
node in the path is $a^{\sigma(X(a))}$.  The path ends when it reaches
a sink. We write $\mathrm{path}(\sigma)$ for the path determined by
$\sigma$.
When extending the path from a node $a$ to the
node $a^{\sigma(X(a))}$, we
say that the \emph{answer to the query} $X(a)$ at~$a$ is $\sigma(X(a))$ and that the path
\emph{sets}
the variable
$X(a)$ to
the value
$\sigma(X(a))$.
For a node $a$ of $\mathrm{path}(\sigma)$, let $\beta(\sigma,a)$ be the restriction of $\sigma$ to the variables that are 
queried in $\mathrm{path}(\sigma)$ in the segment of the path that goes from the 
source to $a$.
For
each node $a$ of the branching program,
let $\beta(a)$ be the maximal partial assignment
that is contained in every $\beta(\sigma,a)$ for all $\sigma$ such that $\mathrm{path}(\sigma)$ 
passes through $a$.
Equivalently, this is the set of all those assignments $x_i \mapsto \gamma$ for
which the query $x_i$ is made, and answered by~$\gamma$, along every
consistent path from the source to $a$. If the program is read-once,
the consistency condition becomes redundant.

The \emph{falsified clause search problem} for an unsatisfiable CNF
formula $F$ is the task of finding a clause $\clc \in F$ that is
falsified by a given truth value assignment~$\sigma$.
A branching program~$P$ on the variables $\vars F$ \emph{solves} the falsified clause search
problem for $F$ if each sink is labelled by a clause of~$F$ such that for
every assignment~$\sigma$, the clause that labels the sink
reached by $\mathrm{path}(\sigma)$ is falsified by~$\sigma$.
The minimal size of any  regular
resolution refutation of an unsatisfiable CNF formula~$F$ is exactly
the same as the minimal size of any
read-once branching program solving the falsified clause
search problem for $F$.
This can be seen by taking the refutation DAG and reversing the edges
to get a branching program or vice versa.
For a formal proof see, e.g., \cite[Theorem  4.3]{Krajicek96BoundedArithmetic}.

\paragraph{The $k$-clique formula}
In order to analyse the complexity of resolution proofs that establish that a
given graph does not contain a $k$-clique we must formulate the
problem as a propositional formula in conjunctive
normal form (CNF).
We consider two
distinct encodings for the clique problem originally defined
in \cite{BIS07IndependentSets}.

\smallskip
The first propositional encoding we present, $\clique{G}{k}$, is based on mapping of vertices to clique members.
This formula is defined
over variables $\icliquemembervar{v}{i}\ (v\in V, i\in [k])$ and consists of the  following set of clauses:
\begin{subequations}
  \begin{align}	
    &\lnot \icliquemembervar{u}{i}\lor \lnot \icliquemembervar{v}{j}
    &&
       i,j \in [\cliquesize], i \neq j,
       u,v \in \Vp, \{u,v\} \notin \Ep \eqcomma
	\label{eq:clique_edge}
    \\[1em]
    &\bigvee_{v\in \Vp}\icliquemembervar{v}{i}
    &&
    i\in [\cliquesize]\eqcomma
    \label{eq:clique_defined}
    \\
    &\lnot \icliquemembervar{u}{i}\lor \lnot \icliquemembervar{v}{i}
    &&
       i \in [\cliquesize],
       u,v \in \Vp, u \neq v \eqcomma
	\label{eq:functionality-axiom}
	\end{align}
\end{subequations}
We refer to
\eqref{eq:clique_edge} as \emph{\edgeaxioms{}},
\eqref{eq:clique_defined} as \emph{\cliqueaxioms{}} and
\eqref{eq:functionality-axiom}
as \emph{\funaxioms{}}.
Note that $\clique{G}{k}$ is satisfiable if and only if $G$ contains a
$k$-clique, and that this is true even if
clauses~\eqref{eq:functionality-axiom}
are omitted---we
write
$\wclique{G}{k}$ to denote this formula with
only clauses~\eqref{eq:clique_edge} and~\eqref{eq:clique_defined}.

\newboolean{extraencodings}
\setboolean{extraencodings}{false}

\ifthenelse{\boolean{extraencodings}}
{
A strengthening of the $\clique{G}{k}$ formula considered in \cite{BIS07IndependentSets} is the \introduceterm{$\cliquemap{G}{k}$} formula that is the conjunction of $\clique{G}{k}$ and the \emph{\orderaxioms{}}
\begin{align}
 &\lnot \icliquemembervar{u}{i+1}\lor \lnot \icliquemembervar{v}{i}
    &&
    u,v \in \Vp, u < v \eqperiod
    \label{eq:order_axiom}	
\end{align}

Notice that, trivially, the size of a minimum (regular) resolution refutation of $\wclique{\Gp}{k}$ is
	bounded from below by the size of a minimum (regular) resolution refutation of
	$\clique{\Gp}{k}$ which in turn is bounded from below by the size of a minimum (regular) resolution refutation of
	$\cliquemap{\Gp}{k}$.

The second encoding is based on counting the number of
vertices in the clique.
The formula, which we denote \introduceterm{$\cliquecount{\Gp}{k}$}, is defined over variables
$\cliquemembervar{v}$ and $\countingvar{v}{i}$ and is encoded
in the following way (where for simplicity we use equivalence $\leftrightarrow$
as a short hand to denote the set of clauses encoding this
equivalence---see~\cite{BIS07IndependentSets} for the full details on the
CNF representation):
\begin{subequations}
	\begin{align}	
		&\lnot \cliquemembervar{u}\lor \lnot \cliquemembervar{v}
		&&
		u,v \in \Vp, \{u,v\} \notin \Ep \eqcomma
		\label{eq:clique_edge1}
		\\
		& \countingvar{v}{0} \leftrightarrow (\countingvar{v-1}{0} \land \lnot \cliquemembervar{v})
		&&
		v \in \Vp \eqcomma
		\label{eq:counting_a}
		\\
		& \countingvar{v}{i} \leftrightarrow ((\countingvar{v-1}{i} \land \lnot \cliquemembervar{v}) \lor (\countingvar{v-1}{i-1} \land \cliquemembervar{v}))
		\ifthenelse{\boolean{conferenceversion}}{\hspace{-1cm} \nonumber  \\
		&}{}&&
		i \in [\cliquesize],
		v \in \Vp, i < v \eqcomma
		\label{eq:counting_b}
		\\
		& \countingvar{i}{i} \leftrightarrow  (\countingvar{i-1}{i-1} \land \cliquemembervar{i})
		&&
		i \in [\cliquesize]
		\eqcomma
		\label{eq:counting_c}
		\\
		& \countingvar{\np}{k}  \eqperiod
		&&
		\label{eq:k_members}
	\end{align}
\end{subequations}

The third and final
}
{\smallskip
	The second}
version of clique formulas that we consider is the block encoding $\cliqueblock{G}{k}$. This formula differs from the previous ones in that it requires a $k$-clique that has a certain ``block-respecting'' structure.
Let $\Vp_1 \dot\cup \Vp_2  \dot\cup \cdots \dot\cup \Vp_k = \Vp$ be a \emph{balanced} $k$-partition of~$\Vp$, that is a partition of $V$ into $k$ disjoint sets each of them of size at most one integer away from $\frac{\lvert V\rvert}{k}$.
The formula $\cliqueblock{G}{k}$, defined
over variables $\cliquemembervar{v}$, encodes the fact that the graph
contains a \introduceterm{transversal $k$-clique}, that is, a $k$-clique in which each clique member belongs to a different block.
Formally, for any positive $k$ and any graph $G$,
the formula
$\cliqueblock{\Gp}{k}$ consists of the following set of clauses:
\begin{subequations}
	\begin{align}	
		&\lnot \cliquemembervar{u}\lor \lnot \cliquemembervar{v}
		&&
		u,v \in \Vp, u\neq v, \{u,v\} \notin \Ep \eqcomma
		\label{eq:clique_edge3}
		\\[1em]
		&\bigvee_{v\in \Vp_i} \cliquemembervar{v}
		&&
		i\in [\cliquesize]\eqcomma
		\label{eq:clique_defined3}
		\\
		&\lnot \cliquemembervar{u} \lor \lnot \cliquemembervar{ v}
		&&
		i \in [\cliquesize],
		u,v \in \Vp_i, u \neq v  \eqperiod
		\label{eq:functionality_axiom3}
	\end{align}
\end{subequations}
\ifthenelse{\boolean{conferenceversion}}{}{
We refer to
\eqref{eq:clique_edge3} as \emph{\edgeaxioms{}},
\eqref{eq:clique_defined3} as \emph{\cliqueaxioms{}},
and
\eqref{eq:functionality_axiom3}
as \emph{\funaxioms{}}.}

Note that a graph can contain a $k$-clique but contain no
transversal $k$-clique for a given partition. Intuitively it is clear that proving that a graph does
not contain a transversal $k$-clique should be easier than
proving it does not contain any $k$-clique, since any proof of the
latter fact must in particular establish the former. We make this
intuition formal below.

\begin{lemma}[\cite{BIS07IndependentSets}]
\label{lem:block-map}
	For any graph $G$ and any $k\in \Nplus$,
	the size of a minimum regular resolution refutation of $\clique{\Gp}{k}$ is
	bounded from below by the size of a minimum regular resolution refutation of
	$\cliqueblock{\Gp}{k}$.
\end{lemma}

This lemma was proven in~\cite{BIS07IndependentSets} for tree-like and for general resolution via a restriction argument, and it is straightforward to see that the same proof holds for regular resolution as well.
\section{Graphs That Are Easy for Regular Resolution}
\label{sec:shortproofs}

Before proving our main $n^{\bigomega{k}}$ lower bound,
in this section
we exhibit classes of graphs
whose clique formulas have
regular resolution refutations
of fixed-parameter tractable length, i.e., length $f(k)\cdot n^{O(1)}$
for some function $f$.  This illustrates the strength of regular
resolution for the $k$-clique problem.  We note that the upper bounds
claimed in this section hold not only for $\clique{\Gp}{k}$ but even for
the subformula $\wclique{\Gp}{\cliquesize}$ that omits the
\funaxioms~\eqref{eq:functionality-axiom}. %

The first example is the class of $(\cliquesize-1)$-colourable graphs.
Such graphs are
hard for tree-like resolution~\cite{BGL13ParameterizedDPLL},
and the known algorithms that distinguish
them from graphs
that contain $\cliquesize$-cliques are highly non-trivial
\cite{Lovasz79ShannonCapacity,Knuth94SandwichTheorem}.
The second example is
the class of graphs that have
a homomorphism into a fixed $\cliquesize$-clique free
graph. 

Recall that a homomorphism from a graph $G = (V,E)$ into a
graph $G' = (V',E')$ is a mapping $h : V \rightarrow V'$ that maps
edges $\{u,v\} \in E$ into edges $\{h(u),h(v)\} \in E'$.
A graph is $(\cliquesize-1)$-colourable if and only if it has
a homomorphism into the $(\cliquesize-1)$-clique, which is of course
$k$-clique free. Therefore our second example is a generalization of
the first one (but the function $f(k)$ becomes larger).

Both upper bounds follows from a generic procedure, based on
Algorithm~\ref{alg:search}, that builds read-once branching programs
for the falsified clause search problem for
$\wclique{G}{\cliquesize}$.

  Given a $\cliquesize$-clique free graph $G$ define
  \begin{equation}
    I(G) = \Set{G\big[\neighcV{R}\big] \;:\;\text{$R$ is a clique in $G$} }\eqperiod
  \end{equation}
  
  \begin{proposition}
    \label{stm:algsearch}
  There is an efficiently constructible read-once branching program
  for the falsified clause search problem on formula
  $\wclique{G}{\cliquesize}$ of size at most
  $\setsize{I(G)}\cdot \cliquesize^{2} \cdot\setsize{V(G)}^{2}$.
\end{proposition}
\begin{proof}
  We build the branching program recursively, following the strategy
  laid out by Algorithm~\ref{alg:search}. For the base case
  $\cliquesize=1$, $G$ must be the graph with no vertices.
  The branching program is a single sink node that outputs the clique
  axiom of index $1$, \ie the empty clause.

  For $\cliquesize>1$, fix $n=\setsize{V(G)}$ and an ordering $v_{1}, \ldots, v_{n}$
  of the vertices in $V(G)$.
  We first build a decision tree $T$ by querying the variables
  $\icliquemembervar{v_{1}}{\cliquesize},
  \icliquemembervar{v_{2}}{\cliquesize}, \ldots$ in order, until
  we get an answer $1$, or until all variables
  with second index $\cliquesize$ have been queried.
  If $\icliquemembervar{v_{j}}{\cliquesize}=0$ for all $j \in [n]$ then
  the $\cliquesize$th clique axiom~\eqref{eq:clique_defined} is
  falsified by the assignment (see line~\ref{line:failure}).
  Otherwise, let $v$ be the first vertex in the order where
  $\icliquemembervar{v}{\cliquesize}=1$. The decision tree now queries
  $\icliquemembervar{w}{i}$ for all $w \in V(G)\setminus\neigh{v}$ and all
  $i<\cliquesize$ to check whether an edge axiom involving~$v$ is
  falsified (lines~\ref{line:checkA}--\ref{line:checkB}).
  If any of these variables is set to $1$ the branching stops and the leaf
  node is labelled with the corresponding edge axiom
  $\neg\icliquemembervar{v}{\cliquesize} \vee
  \neg\icliquemembervar{w}{i}$.

\begin{algorithm}[t]
\ifthenelse{\boolean{conferenceversion}}{
{\bf Input} $\cliquesize \in \Nplus$, a $\cliquesize$-clique free graph $G$, an assignment $\alpha\!:\!{\{\icliquemembervar{v}{i}\text{ for }v \in V(G), i \in [\cliquesize]\}}\rightarrow{\{0,1\}}$

{\bf Output} A clause of $\wclique{G}{\cliquesize}$ falsified by $\alpha$

\begin{algorithmic}[1]
\Procedure{Search}{$G,\cliquesize,\alpha$}
	\For{$v\in V(G)$}
		 \If{$\alpha(\icliquemembervar{v}{\cliquesize})=1$}
		 	\For{$w \not\in\neigh{v}$ and $i<\cliquesize$} \label{line:checkA}
		 		\If{$\alpha(\icliquemembervar{w}{i})=1$}
		 			\State \Return edge axiom $\neg\icliquemembervar{v}{\cliquesize}
                \vee \neg\icliquemembervar{w}{i}$ ~\eqref{eq:clique_edge}. \label{line:checkB}
		 		\EndIf
		 	\EndFor
		 	\State $G'   \leftarrow G[\neigh{v}]$ \label{line:recurseA}
		 	\State $\alpha' \leftarrow \alpha$ restricted to variables $\icliquemembervar{w}{j}$ for $w\in V(G')$ and $1\leq j \leq \cliquesize-1$
		 	\State \Return \textsc{Search}$(G',\cliquesize-1,\alpha')$\label{line:recurseB}
		 \EndIf
	\EndFor
	\State \Return the $\cliquesize$th clique axiom~\eqref{eq:clique_defined}.\label{line:failure}
\EndProcedure
\end{algorithmic}
}
{
  \begin{algorithm2e}[H]
    \SetKwInOut{Input}{Input}
    \SetKwInOut{Output}{Output}
    \SetKwInOut{Dummy}{Dummy} %
    \Input{~$\cliquesize \in \Nplus$, a $\cliquesize$-clique free graph $G$, an assignment
      $\alpha\!:\!{\{\icliquemembervar{v}{i}\text{ for }v \in V(G),
        i \in [\cliquesize]\}}\rightarrow{\{0,1\}}$}
    \Output{~A clause of $\wclique{G}{\cliquesize}$ falsified by $\alpha$}
    $\mathtt{Search}(G,\cliquesize,\alpha)$:
    \Begin{
      \For{$v \in V(G)$}{
        \If{$\alpha(\icliquemembervar{v}{\cliquesize})=1$}
        {\For{$w \in V(G)\setminus\neigh{v} $ and $i<\cliquesize$}{\label{line:checkA}
            \lIf{$\alpha(\icliquemembervar{w}{i})=1$}
            {\Return{edge axiom $\neg\icliquemembervar{v}{\cliquesize}
                \vee \neg\icliquemembervar{w}{i}$ ~\eqref{eq:clique_edge} }}}\label{line:checkB}
          $G'   \leftarrow G[\neigh{v}]$\;\label{line:recurseA}
          $\alpha' \leftarrow \alpha$ restricted to variables $\icliquemembervar{w}{j}$ for $w\in V(G')$ and $1\leq j \leq \cliquesize-1$\;
          \Return{$\mathtt{Search}(G',\cliquesize-1,\alpha')$}\label{line:recurseB}
        }
      }
      \Return{the $\cliquesize$th clique axiom~\eqref{eq:clique_defined} }\label{line:failure}
      }
  \end{algorithm2e}
}

  \caption{Read-once branching program for the falsified clause search problem on $\wclique{G}{\cliquesize}$.}
  \label{alg:search}
\end{algorithm}

  The decision tree $T$ built so far has at most $\cliquesize n^{2}$
  nodes, and we can identify $n$ ``open'' leaf nodes
  $a_{v_{1}}, a_{v_{2}}, \ldots, a_{v_{n}}$, where $a_{v_{i}}$ is the leaf node
  reached by the path that sets $\icliquemembervar{v_{i}}{\cliquesize}=1$ and that does
  not yet determine the answer to the search problem.
  Let us focus on a specific node $a_{v}$ for some $v \in V(G)$. %
  The partial assignment  $\mathrm{path}(a_{v})$ sets $v$ to be the $\cliquesize$th
  member of the clique and every vertex in $V(G)\setminus\neigh{v}$ to
  not be in the clique.
  Let $G_{v}$ be the subgraph induced on~$G$ by $\neigh{v}$,
  let $S_{v}$ be the set of variables
  $\icliquemembervar{w}{i}$ for $w \in\neigh{v}$ and $i<\cliquesize$,
  and let
  $\rho_{v}$ be the partial assignment setting
  $\icliquemembervar{w}{i}=0$ for $w \in V(G)\setminus\neigh{v}$ and
  $i<\cliquesize$. Clearly $\rho_{v} \subseteq \mathrm{path}(a_{v})$.

  By the inductive hypothesis there exists a branching program~$B_{v}$ that
  solves the search problem on $\wclique{G_{v}}{\cliquesize-1}$
  querying only variables in $S_{v}$.
  This corresponds to the recursive call for the subgraph $G_{v}$ and
  $\cliquesize-1$ (lines~\ref{line:recurseA}--\ref{line:recurseB}).
  If we attach each $B_{v}$ to $a_{v}$ we get a complete branching
  program for $\wclique{G}{\cliquesize}$. This is read-once because
  $B_{v}$ only queries variables in $S_{v}$ and these variables are not in
  $\mathrm{path}(a_{v})$.

  To prove that the composed program is correct we consider an
  assignment $\sigma$ to the variables in $S_{v}$ and show that the clause
  output by $B_{v}$ on $\sigma$ is also a valid output for the search
  problem on $\wclique{G}{\cliquesize}$, \ie it is falsified by the
  assignment $\mathrm{path}(a_{v}) \cup \sigma$.
  Actually we show the stronger claim that it is falsified by
  $\rho_{v} \cup \sigma$, which is a subset of $\mathrm{path}(a_{v}) \cup \sigma$.
  To this end, note that if the output of $B_{v}$ on $\sigma$ is an edge axiom of
  $\wclique{G_{v}}{\cliquesize-1}$, this must be some
  $\neg \icliquemembervar{u}{i} \vee \neg\icliquemembervar{w}{j}$ for
  $i,j < \cliquesize$, which is also an edge axiom of
  $\wclique{G}{\cliquesize}$ and is falsified by
  $\sigma$. %
  Now if the output of $B_{v}$ on $\sigma$ is the $i$th clique axiom of
  $\wclique{G_{v}}{\cliquesize-1}$, then
  $\sigma$ falsifies $\bigvee_{w\in\neigh{v}} \icliquemembervar{v}{i}$,
  and therefore $\rho_{v} \cup \sigma$ falsifies the $i$th clique
  axiom 
  of 
  $\wclique{G}{\cliquesize}$.

  The construction so far is correct but produces a very large
  branching program (in particular it has tree-like structure on top).
  In order to create a smaller branching program, we observe that if
  $u,v\in V(G)$ are such that $\neigh{u}=\neigh{w}$ then
  $G_{u}=G_{w}$, $B_{u}=B_{w}$ and $\rho_{u}=\rho_{w}$.
  This observation allows us to merge together all nodes $a_{v}$ that
  have the same value of $\neigh{v}$ into a single node, and to
  identify all the corresponding copies of the same branching program
  $B_{v}$.
  Now let us focus on some node $a^{*}$ obtained by this merge process,
  and pick arbitrarily some $a_{v}$ that was merged into it (the
  specific choice is irrelevant).
  By construction $\rho_{v}$ is consistent with all paths reaching
  $a^{*}$, but we can claim further: $\rho_{v}$ is consistent with all
  paths \emph{passing through} $a^{*}$ because $B_{v}$ only queries
  variables in $S_{v}$, which is disjoint from the domain of
  $\rho_{v}$.
  Because of this last fact all paths that pass through node $a^{*}$
  and reach an output node $b^{*}$ in the attached copy of $B_{v}$
  must contain the partial assignment $\rho_{v} \cup \sigma$, where
  $\sigma$ is the common partial assignment consistent with all paths
  from the root of $B_{v}$ to $b^{*}$.
  If $b^{*}$ outputs an edge axiom, this is already falsified by
  $\sigma$ because of the correctness of $B_{v}$. If $b^{*}$ outputs
  the $i$th clique axiom, the correctness of $B_{v}$ guarantees that
  $\sigma$ falsifies the $i$th axiom for $G_{v}$, and therefore
  $\rho_{v} \cup \sigma$ falsifies the $i$th clique axiom of $G$.
  Hence the new branching program is correct.
  
  This merge process leads to having only one subprogram for each distinct
  induced subgraph at each level of the recursion.
  In order to bound the size of this program, we decompose it into~$k$ levels.
  The source is at level zero and corresponds to the
  graph $G$. At level $i$ there are nodes corresponding to all
  subgraphs induced by the common neighbourhood of cliques of size $i$.
  Each node in the $i$th level connects to the nodes of the $(i+1)$th
  level by a branching program of size at most $\cliquesize n^{2}$.
  Notice that an induced subgraph in $I(G)$ cannot occur twice in the
  same layers, so the total size of the final branching program is
  at most $\setsize{I(G)}\cdot \cliquesize^{2} n^{2}$ nodes.
\end{proof}

We now proceed to prove the upper bounds mentioned previously. A graph
$G$ that has a homomorphism into a small $\cliquesize$-clique free
graph $H$ may still have a large set $I(G)$, making
Proposition~\ref{stm:algsearch} inefficient.
The first key observation is that if~$G$ has a homomorphism
into a graph~$H$ then it is a subgraph of a blown up version of~$H$,
namely, of a graph obtained by transforming
each vertex of~$H$ into a ``cloud'' of vertices where a cloud does
not contain any edge, two clouds corresponding to two adjacent
vertices in~$H$ have all possible edges between them, and two clouds
corresponding to two non-adjacent vertices in~$H$ have no edges between
them.
A second crucial point is that if~$G'$ is a blown up version of~$H$
then it turns out that $\setsize{I(G')}=\setsize{I(H)}$, making
Proposition~\ref{stm:algsearch} effective for~$G'$.
The upper bound then follows from observing that the task of proving that~$G$ is
$\cliquesize$-clique free should not be harder than the same task for
a supergraph of~$G$. Indeed Fact~\ref{fact:subgraph} formalises this intuition.
It is interesting to observe that the constructions in
Proposition~\ref{stm:algsearch} and in Fact~\ref{fact:subgraph} are
efficient. The non-constructive part is guessing the %
homomorphism to~$H$.

\begin{fact}
  \label{fact:subgraph}
  Let $G=(V,E)$ and $G'=(V',E')$ be graphs with no $\cliquesize$-clique
  such that $V\subseteq V'$ and
  \mbox{$E \subseteq E' \cap \binom{V}{2}$}.
  If $\wclique{G'}{\cliquesize}$ has a (regular) refutation of \lengthsize~$\prooflength$, then
  $\wclique{G}{\cliquesize}$ has a (regular) refutation of \lengthsize at most~$\prooflength$.
\end{fact}
\begin{proof}
  Consider the partial assignment $\rho$ that sets
  $\icliquemembervar{v}{i}=0$ for every $v \not\in V$ and
  $i\in[\cliquesize]$.
  The restricted formula $\restrict{\wclique{G'}{\cliquesize}}{\rho}$
  is isomorphic to $\wclique{\widetilde{G}}{\cliquesize}$, where
  $V(\widetilde{G})=V$ and $E(\widetilde{G})=E' \cap \binom{V}{2}$, and thus,
  by Fact~\ref{fact:restrict}, has a (regular) refutation $\proofstd$
  of \lengthsize at most~$\prooflength$.
  Removing edges from a graph only introduces additional edge
  axioms~\eqref{eq:clique_edge} in the corresponding formula,
  therefore
  $\wclique{\widetilde{G}}{\cliquesize} \subseteq \wclique{G}{\cliquesize}$
  and $\proofstd$ is a valid refutation of $\wclique{G}{\cliquesize}$ as well.
\end{proof}

It was shown in
~\cite{BGL13ParameterizedDPLL}
that the $\cliquesize$-clique
formula of a complete $(\cliquesize-1)$-partite graph on $n$ vertices has
a regular resolution refutation of \lengthsize $2^{\cliquesize} n^{O(1)}$,
although the regularity is not stressed in that paper.
Since it is instructive to see how this refutation is constructed in this
framework, we give a self-contained proof.

\begin{proposition}[{\cite[Proposition 5.3]{BGL13ParameterizedDPLL}}]
  \label{stm:upperbound_colourable}
  If $G$ is a $(\cliquesize-1)$-colourable graph on $n$ vertices, then
  $\wclique{G}{\cliquesize}$ has a regular resolution refutation of
  \lengthsize at most $2^{\cliquesize} \cliquesize^{2} n^{2}$.
\end{proposition}
\begin{proof}
  Let $V=V(G)$ and let
  $V_{1} \dot\cup V_{2} \dot\cup \dots \dot\cup V_{(\cliquesize-1)}$
  be a partition of $V$ into colour classes.
  Define the graph $G'=(V,E')$ where the edge set~$E'$ has an edge between
  any pair of vertices belonging to two different colour classes.
  Clearly~$G$ is a subgraph of~$G'$.
  Observe that any clique~$R$ in~$G'$ has at most one vertex
  in each colour class, and that the common neighbours of~$R$ are all
  the vertices in the colour classes not touched by~$R$.

  Therefore, there is a one-to-one correspondence between the members
  of $I(G')$ and the subsets of $[\cliquesize-1]$.
  By Proposition~\ref{stm:algsearch} there is a read-once branching program
  for the falsified clause search problem on formula
  $\wclique{G'}{\cliquesize}$ of size at most
  $2^{\cliquesize} \cliquesize^{2} n^{2}$.
  This read-once branching program corresponds to a regular resolution
  refutation of $\wclique{G'}{\cliquesize}$ of the same size.
  By Fact~\ref{fact:subgraph} there must be a regular resolution
  refutation of size at most
  $2^{\cliquesize} \cliquesize^{2} n^{2}$ for
  $\wclique{G}{\cliquesize}$ as well.
\end{proof}

Next we generalize Proposition~\ref{stm:upperbound_colourable} to
graphs~$G$ that have a homomorphism to a $k$-clique free graph~$H$.

\begin{proposition}
  \label{stm:upperbound_homomorsphism}
  If $G$ is a graph on $n$ vertices that has a homomorphism into a $\cliquesize$-clique free
  graph $H$ on $m$ vertices, then $\wclique{G}{\cliquesize}$ has a regular resolution
  refutation of \lengthsize at most $m^{\cliquesize} \cliquesize^{2} n^{2}$.
\end{proposition}
\begin{proof}
  Fix a homomorphism $h\!:\! V(G) \rightarrow V(H)$ and an ordering
  $u_{1},\ldots,u_{m}$ of the vertices of $H$.
  Let $V_{1} \dot\cup V_{2} \dot\cup \dots \dot\cup V_{m}$
  be the partition of $V(G)$ such that $V_{i}$ is the set of
  vertices of $G$ mapped to $u_{i}$ by~$h$.
  We define the graph $G'=(V,E')$ where
  \begin{equation}
    E' = \!\!\!\!\!\!\!\! \bigcup_{\{u_{i},u_{j}\}\in E(H)} \!\!\!\!\!\!\!\!  V_{i} \times V_{j} \eqcomma
  \end{equation}
  that is, $G'$ is a blown up version of $H$ that contains $G$ as a subgraph.
  To prove our result
  we note that, by Proposition~\ref{stm:algsearch}, there is a read-once
  branching program for the falsified clause search problem on
  $\wclique{G'}{\cliquesize}$---and hence also a regular resolution
  refutations of the same formula---of size at most $\setsize{I(G')}\cdot \cliquesize^2 n^2$.
  This implies that, by
  Fact~\ref{fact:subgraph}, there is a regular resolution refutation of
  $\wclique{G}{\cliquesize}$ of at most the same size.

  To conclude the proof it remains only to show that $\setsize{I(G')}\leq m^k$.
  By construction, $h$ maps injectively a clique $R \subseteq V(G')$ into
  a clique $R_H \subseteq V(H)$ of the same size.
  Moreover, note that if $U=\neighcV{R_{H}}$,
  then $\neighcV{R}=\cup_{u_i \in U}V_{i}$.
  Therefore, for any clique $R'\subseteq V(G')$ that is mapped by $h$ to $R_H$ it holds that $\neighcV{R}=\neighcV{R'}$, \ie $\neighcV{R'}$ is completely characterized by the clique in $H$ it is mapped to.
  Thus $I(G)$ has at most one
  element for each clique in $H$ and we have that $\setsize{I(G')} = \setsize{I(H)}$. Finally, note that $\setsize{I(H)} \leq m^{\cliquesize}$ since, being $\cliquesize$-clique free, $H$
  cannot have more than $\sum_{i=0}^{\cliquesize-1} m^i  \leq m^{\cliquesize}$ cliques.
\end{proof}
\ifthenelse{\boolean{conferenceversion}}
{\section{Random Graphs Are Hard}}
{\section{Random Graphs Are Hard for Regular Resolution}}
\label{sec:recap-lower-bounds}

The main result of this paper is an
average case lower bound of $\np^{\bigomega{k}}$ for regular resolution
for the $k$-clique problem. As we
saw in Section~\ref{sec:preliminaries}, the $k$-clique problem can
be encoded in different ways and depending on the preferred formula
the range of~$k$ for which we can obtain a lower bound differs.
In this section we present a summary of our results for the different encodings.

\begin{theorem}
	\label{thm:k-clique-erdos-renyi-block}
        For any real constant $\epsilon > 0$,
        any sufficiently large integer~$n$,
        any positive integer~$k\leq {n^{1/4-\epsilon}}$,
        and any real $\oldeta>1$,\,
        if\, $\sample{\randomvarG}{\randG{\np}{\np^{-2\oldeta/(k-1)}}}$ is an
        \ErdosRenyi random graph, then,
        with probability at least $1-\exp(-\sqrt{\np})$,
        any regular resolution refutation of
        $\cliqueblock{\randomvarG}{k}$ has \lengthsize at
        least~$\np^{\Omega(k/{\oldeta^2})}$.
\end{theorem}

The parameter $\oldeta$ determines the density of the graph:
the larger~$\oldeta$ the sparser the graph and the problem
of determining whether~$\randomvarG$ contains a $k$-clique becomes
easier. 
For constant $\oldeta$, the edge probability implies the graph $\randomvarG$ has clique number concentrated around $k/\oldeta$ and the theorem yields
a $\np^{\bigomega{k}}$ lower bound which is tight up to the multiplicative
constant in the exponent. 
The lower bound decreases smoothly with the
edge density and is non-trivial for $\oldeta = o(\sqrt{k})$.

A problem which is closely related to the problem we consider is that of
distinguishing a random graph sampled from $\randG{\np}{p}$ from
a random graph from the same distribution with a planted $k$-clique.
The most studied setting is when $p=1/2$. In this scenario the problem
can be solved in polynomial time with high probability for
$k\approx \sqrt{\np}$~\cite{Kucera95ExpectedComplexity,AKS98Finding}.
It is still an open problem whether there exists a
polynomial time algorithm solving this problem for
$\log \np \ll k\ll \sqrt{n}$. %
For $\sample{\randomvarG}{\randG{\np}{1/2}}$, setting $\xi = 
\frac{k-1}{2\log_2(n)}$, Theorem~\ref{thm:k-clique-erdos-renyi-block} implies that to refute
$\cliqueblock{\randomvarG}{k}$
\aas regular resolution requires $n^{\Omega(\log^2(n)/k)}$ size; which is $\np^{\bigomega{\log n}}$  size
for $k= O(\log n)$ and super-polynomial size
for $k=\littleoh{\log^2 n}$.
We note that, in the case $k = O(\log n)$, the lower bound is tight. 
This follows from Proposition~\ref{stm:algsearch}
since \aas there are at most $n^{O(\log n)}$ different 
cliques in $\sample{\randomvarG}{\randG{\np}{1/2}}$
(because \aas the largest clique has size 
at most $2\log n$) and, therefore, the set
${I(\randomvarG)}$ in Proposition~\ref{stm:algsearch} has size
at most $n^{O(\log n)}$.

An interesting question is whether \refthm{thm:k-clique-erdos-renyi-block} holds for larger values of $k$.
We show that for the formula $\clique{\Gp}{k}$ (recall that by Lemma \ref{lem:block-map} this encoding is easier
for the purpose of lower bounds) we can prove the lower bound
for $k\leq {n^{1/2-\epsilon}}$ as long as the edge density of the graph is close
to the threshold for containing a $k$-clique.

\begin{theorem}
	\label{thm:k-clique-erdos-renyi}
        For any real constant $\epsilon > 0$,
        any sufficiently large integer~$n$,
        any positive integer~$k$,
        and any real $\oldeta>1$ such that $k \sqrt{{\oldeta}} \leq  {n^{1/2-\epsilon}}$,
        if $\sample{\randomvarG}{\randG{\np}{\np^{-2\oldeta/(k-1)}}}$ is an
        \ErdosRenyi random graph, then,
        with probability at least $1-\exp(-\sqrt{\np})$,
        any regular resolution refutation of
        $\clique{\randomvarG}{k}$ has \lengthsize at
        least~$\np^{\Omega(k/{\oldeta^2})}$.
\end{theorem}

\ifthenelse{\boolean{conferenceversion}}{
In this extended abstract we prove \refthm{thm:k-clique-erdos-renyi} and we refer to the upcomming full-length version of this paper for the proof of \refthm{thm:k-clique-erdos-renyi-block}.
}{
In this paper we prove \refthm{thm:k-clique-erdos-renyi-block} and we refer to the conference version of this paper~\cite{ABdRLNR18Clique} for the proof of \refthm{thm:k-clique-erdos-renyi}.
}
We note, however, that both proofs are very similar and having seen one it is an easy exercise to obtain the other.
The proof of
\ifthenelse{\boolean{conferenceversion}}{\refthm{thm:k-clique-erdos-renyi}}{\refthm{thm:k-clique-erdos-renyi-block}} is deferred to Section~\ref{sec:lowerbound2} and is based on a general lower bound technique we develop in Section~\ref{sec:lowerbound1}.
\section{Clique-Denseness Implies Hardness for Regular Resolution}
\label{sec:lowerbound1}

\newcommand{\parameterspec}[2]{#2\xspace}

In this section we define a combinatorial property of graphs,
which we call \emph{\cliquedensenessNOparam{}},
and prove that if a \mbox{$k$-clique}-free graph~$\Gp$ is
\cliquedenseNOparam with the appropriate parameters, then this implies a
lower bound~$n^{\bigomega{k}}$ on the \lengthsize of any regular resolution
refutation of the $k$-clique formula on~$\Gp$.

In order to argue that regular resolution has a hard time certifying
the \mbox{$k$-clique}-freeness of a  graph~$\Gp$, one property
that seems useful to have is that
for every small enough clique in the graph there are many ways of
extending it to a larger clique.
In other words, if $\Rp\subseteq \Vp$ forms a clique and $\Rp$
is small, we would like
the common neighbourhood~$\commonneighbourhood{\Rp}{\Vp}$ to be large. This
motivates the following definitions.

\begin{definition}[\GenericNOparam set]
  \label{def:genericforR}	
  Given
  \parameterspec
  {a graph $\Gp=(\Vp,\Ep)$}
  {$\Gp=(\Vp,\Ep)$}
  and
  \parameterspec
  {real numbers $\qp$ and $\rp$,}
  {$\qp, \rp \in \Rplus$,}
  a set
  $\Wp\subseteq \Vp$ is \emph{\genericforR{\qp} for $\Rp \subseteq \Vp$}
  if $\Neighcsize{\Rp}{\Wp}\geq \qp$. We say that $\Wp$ is
  \emph{\generic{\rp}{\qp}} if it
  is \genericforR{\qp} for every $\Rp \subseteq \Vp$
  of size~$\setsize{\Rp} \leq \rp$.
\end{definition}

If $\Wp$ is  an \generic{\rp}{\qp} set, then
we know that any clique of size~$\rp$ can be extended to a clique of
size~$\rp+1$ in at least $\qp$ different
ways by adding some vertex of $\Wp$.  Note,
however,
that the definition of \generic{\rp}{\qp} is more general than this
since $\Rp$ is not required to be a clique.

Next we define a more robust notion of
\genericnessNOparam.
For some settings of~$\rp$ and~$\qp$ of interest to us it is too much to hope for a set
$\Wp$ that is \genericforR{\qp} for every $\Rp\subseteq \Vp$ of size
at most~$\rp$. In this case we would still like to be able to find a
``mostly \genericNOparam'' set~$\Wp$ in the sense that we can ``localize''
bad (\ie those for which $\Wp$
fails to be \genericforR{\qp}) sets
$\Rp\subseteq \Vp$ of size $\setsize{\Rp}\leq \rp$.

\begin{definition}[\RobgenericNOparam set]
  \label{def:robgeneric}	
  Given
  \parameterspec
  {a graph $\Gp=(\Vp,\Ep)$}
  {$\Gp=(\Vp,\Ep)$}
  and
  \parameterspec
  {real numbers $\largerp,\rp,\qp'$ and $\spar$}
  {$\largerp,\rp,\qp,\spar \in \Rplus$}
  with $ \largerp\geq \rp$, a set $\Wp\subseteq \Vp$
  is \emph{ \robgeneric{\largerp}{\rp}{\qp}{\spar}}
  if there exists a
  set $S\subseteq \Vp$ of size~$\setsize{\Sp}\leq s$
  such that for every $\Rp\subseteq \Vp$ with
  $\setsize{\Rp}\leq\largerp$
  for which
  $\Wp$ is not \genericforR{\qp}, it holds that
  $\setsize{\Rp \cap \Sp} \geq \rp$.
\end{definition}

In what follows, it might be helpful for the reader to think of $\largerp$
and $\rp$ as linear in~$k$, and $\qp$ and~$\spar$ as
polynomial in $\np$,
where we also have $\spar\ll \qp$.

Now we are ready to define a property of graphs that makes it hard for
regular resolution to certify that graphs with this property are
indeed $k$-clique-free.

\begin{definition}[\CliquedenseNOparam graph]
  \label{def:propertyP}	
  Given
  {$k \in \Nplus$ and $\tp, \spar, \newepsilon \in \Rplus$, $1\leq\tp\leq k$}
  we say that
  a graph $\Gp=(\Vp,\Ep)$ with a $k$-partition
  $\Vp_1 \cup \cdots \cup \Vp_k = \Vp$ is
  \emph{\cliquedense{k}{\tp}{\rp}{\spar}{\newepsilon}}
  if there exist
  {$ \rp,\qp \in \Rplus$, $\rp \geq 4k/\tp^2$,}
  such that
  \begin{enumerate} \itemsep=0pt
  \item 	
    \label{item:hyp-generic}
    $\Vp_i$ is  \generic{\tp\rp}{\tp\qp} for all $i\in [k]$, and
  \item
    \label{item:hyp-awesome}
    every \generic{\rp}{\qp} set $\Wp\subseteq \Vp$
    is \robgeneric{\tp\rp}{\rp}{\qp'}{\spar}
    for
    $\qp'=\newepsilon\rp\spar^{1+\newepsilon}\log\spar$.
  \end{enumerate}	 	
\end{definition}

\begin{remark}[The complete $(\cliquesize-1)$-partite graph is not \cliquedenseNOparam{}]
  Since the property of \cliquedensenessNOparam{} in
  Definition~\ref{def:propertyP} is a sufficient condition for the
  lower bound, it is worth to pause and observe that this property
  does not hold for examples such as 
  $(\cliquesize-1)$-colourable graphs, which have non-trivially short proofs.
  
  Indeed, consider the
  $(\cliquesize-1)$-colourable graph $\Gp=(\Vp,\Ep)$ with balanced colour classes and
  maximum edge set.
  Namely, $\Vp=\bigcup_{c} U_{c}$ for $c\in[\cliquesize-1]$ and 
  $\setsize{U_{c}}={n/(\cliquesize-1)}$, and
  the edges of $\Gp$ are all pairs $\{u,v\}$ for $u \in U_{c}$ and
  $v \in U_{c'}$ with $c\neq c'$.
  The graph~$\Gp$ satisfies property~\eqref{item:hyp-generic} of clique-denseness for
  any $\cliquesize$-partition of $\Vp$ that splits each colour class
  roughly equally among parts, but 
  fails to satisfy property~\eqref{item:hyp-awesome} in a rather extreme way.
  To see why, fix any integer $\rp < \cliquesize-1$ and let~$W$ be the union of $\rp+1$
  arbitrarily chosen colour classes.
  The set $W$ is \generic{\rp}{\qp} for any $\qp$ up to
  $n/(\cliquesize-1)$, because the common neighbourhood of any $r$
  vertices in~$V$ must contain one of the colour classes
  $U_{c}\subseteq W$.

  Can $W$ be \robgeneric{\tp\rp}{\rp}{\qp'}{\spar} for some choice
  of parameters? First note that $tr \geq r+1$ (since $r \leq k$ implies $t \geq 2$) and 
  that $\commonneighbourhood{\Rp}{\Wp}=\emptyset$
  for any set~$R$ of size $r+1$ that has one vertex from each colour
  class in~$W$.
  So in order for~$W$ to be \robgeneric{\tp\rp}{\rp}{\qp'}{\spar} %
  there should be a set~$S$ of 
  size $\spar \ll \qp' \leq n/(\cliquesize-1)$ that has a large intersection with any such~$R$. 
  This, however, is not possible
  since %
  $S$ cannot completely cover any of the colour
  classes in~$W$ (because $s \ll n/(\cliquesize-1)$) and thus, for any choice of~$S$, 
  there are sets~$R$ completely disjoint from~$S$ for which 
  $\commonneighbourhood{\Rp}{\Wp}=\emptyset$.
\end{remark}

\begin{theorem}
  \label{thm:main}
  Given $k\in \Nplus$ and $\tp, \spar, \newepsilon \in \Rplus$
  if the graph $\Gp=(\Vp,\Ep)$ with balanced $k$-partition
  $\Vp_1 \cup \cdots \cup \Vp_k = \Vp$
  is \cliquedense{k}{\tp}{\rp}{\spar}{\newepsilon}, then
  every regular resolution refutation of the CNF formula
  $\cliqueblock{\Gp}{k}$ has \lengthsize at least $\Bigomega{\spar^{\newepsilon k/\tp^2}}$.
\end{theorem}

The value of $\qp^\prime$ in Definition~\ref{def:propertyP} can be tailored in order to prove Theorem~\ref{thm:k-clique-erdos-renyi-block} for slightly larger values of $k$. For example, setting $\qp^\prime=3\newepsilon \spar^{1+\newepsilon}\log\spar$ and making the necessary modifications in the proof would yield
Theorem~\ref{thm:k-clique-erdos-renyi-block} for $k \ll \np^{1/3}$ but for a smaller range of edge densities.
A similar adjustment was done in the conference
version of this paper~\cite{ABdRLNR18Clique} to obtain Theorem~\ref{thm:k-clique-erdos-renyi} for $k \ll \np^{1/2}$.

We will spend the rest of this section establishing
\refth{thm:main}.
{Fix $\rp,\qp \in  \Rplus$ witnessing that}
$G$ is \cliquedense{k}{\tp}{\rp}{\spar}{\newepsilon}
as per  \refdef{def:propertyP}.
We first note that we can assume that $\tp\rp\leq k$ since otherwise, by
property~\ref{item:hyp-generic} of \refdef{def:propertyP},
$G$ contains a block-respecting $k$-clique and the
theorem follows immediately.

By the discussion in \refsec{sec:preliminaries}
it is sufficient to consider read-once branching
programs, since they are equivalent to regular resolution refutations,
and so in what follows this is the language in which we will phrase
our lower bound.
Thus,
for the rest of this section let $\bprogpi$ be an arbitrary, fixed
read-once branching program that solves the falsified
clause search problem for $\cliqueblock{\Gp}{k}$.
We will use the convention of referring to  ``vertices'' of the
graph~$G$ and  ``nodes'' of the branching program~$\bprogpi$ to distinguish between the two.
We sometime abuse notation and say that a vertex $v\in \Vp$
is set to $0$ or to $1$ when we mean that the corresponding variable~$\cliquemembervar{v}$ is set to $0$ or to $1$.

Recall that for a node $a$ of
$\bprogpi$, $\beta(a)$ denotes the maximal partial assignment
that is contained in every $\beta(\sigma,a)$ for all $\sigma$ such that $\mathrm{path}(\sigma)$ 
passes through $a$, where $\beta(\sigma,a)$ is the restriction of $\sigma$ to the variables that are 
queried in $\mathrm{path}(\sigma)$ in the segment of the path that goes from the 
source to $a$.
For any partial assignment~$\beta$ we write $\betaone$
to denote the partial assignment that contains exactly the variables
that are set to $1$ in $\beta$.
Clearly,
if $\beta$ falsifies
an \edgeaxiom or a \funaxiom,
then so does~$\betaone$.
Furthermore, for any $\gamma \supseteq \beta$,
if $\beta$ falsifies
an axiom
so does~$\gamma$.
We will use this monotonicity property of partial assignments
throughout the proof.

For each node $a$ of
$\bprogpi$ and each index $i \in [k]$
we  define two sets of vertices
  \begin{subequations}
    \begin{align}
\setzeros{i}{a} &= \setdescr{ u \in V_i }{ \beta(a) \text{ sets } \cliquemembervar{u} \text{
		to }  0}
  \\
  \setones{i}{a} &= \setdescr{ u \in V_i }{ \beta(a) \text{ sets } \cliquemembervar{u} \text{
		to }  1}
    \end{align}
  \end{subequations}
of $G$. Observe that for $\beta=\beta(a)$
the set of vertices referenced by variables in $\betaone$ is $\Union_i \setones{i}{a}$.

Intuitively, one can think of $\setzeros{i}{a}$ and~$\setones{i}{a}$
as the only sets of vertices in $\Vp_i$ assigned~$0$ and~$1$,
respectively, that are ``remembered'' at the
node~$a$ (in the language of resolution, they correspond to negative and positive occurrences
of variables in the clause $D_a$ associated with the node $a$).
Other assignments to vertices in $\Vp_i$ encountered along some path
to~$a$ have been ``forgotten'' and may not be queried any more on any path starting at $a$.
Formally, we say that a vertex $v$ is \emph{forgotten at~$a$} if
there is a path from the source of~$\bprogpi$ to~$a$ passing
through a node~$b$ where $v$ is queried, but $v$ is not
in~$\setzeros{i}{a}$ nor in~$\setones{i}{a}$.
Furthermore, we say index~$i$ \emph{is forgotten at~$a$} if
some vertex~$v\in \Vp_i$ is forgotten at $a$.
Of utter importance is the fact that these notions are persistent: if a variable
or an index is forgotten at a node $a$, then it will also be the case for any node
reachable from $a$ by a path.
We say that a path in~$\bprogpi$
\emph{ends in the $i$th  \cliqueaxiom}
if the clause that labels its last node is the
\cliqueaxiom~\eqref{eq:clique_defined3}
of $\cliqueblock{\Gp}{k}$ with index~$i$.
The above observation implies that
the index~$i$ cannot be forgotten at any node along such a path.

We establish our lower bound via a bottleneck counting argument for
paths in~$\bprogpi$.
To this end, let us define a distribution $\Ddist$ over paths in~$\bprogpi$
by the following random process.  The path starts at the source and ends whenever it
reaches a sink of $\bprogpi$.  At an internal node $a$ with successor nodes
$a^0$ and $a^1$, reached by edges labelled $0$ and $1$ respectively,
the process proceeds as follows.
\begin{enumerate}
\itemsep0em
\item \label{item:forced-choice1}
 If $X(a) = \cliquemembervar{u}$ for $u\in \Vp_i$ and $i$ is forgotten at $a$
then the path proceeds via the edge labelled~$0$ to~$a^0$.
\item \label{item:forced-choice2}
 If $X(a) = \cliquemembervar{u}$ and
 $\beta(a) \cup \set{ \cliquemembervar{u} = 1 }$ falsifies
an \edgeaxiom~\eqref{eq:clique_edge3}
or a \funaxiom~\eqref{eq:functionality_axiom3},
then the path proceeds to~$a^0$.
\item \label{item:free-choice}
Otherwise, an independent
$\spar^{-(1+\newepsilon)}$\nobreakdash-biased coin is
tossed with outcome~$\bpbit \in \set{0,1}$ and the random path
proceeds to~$a^{\bpbit}$.
\end{enumerate}
We say that in cases~\ref{item:forced-choice1} and~\ref{item:forced-choice2}
the answer to the query
$X(a)$ is \emph{forced}.
Note that any path $\alpha$ in the support of~$\Ddist$ must end
in a \cliqueaxiom since $\alpha$ does not
falsify any edge or \funaxiom by item~\ref{item:forced-choice2}
in the construction.
Moreover, a property that will be absolutely crucial is that only
answers $0$ can be forced---answers~$1$ are always the result of a
coin flip.

\begin{claim} \label{claim:atmostk}
Every path in the support of $\Ddist$ sets at most $k$ variables
to $1$.
\end{claim}

\begin{proof}
Let $\alpha$ be a path in the support of $\Ddist$.
We argue that for each $i\in [k]$ at most one vertex $u\in \Vp_i$
is such that the variable $\cliquemembervar{u}$ is set to $1$ on $\alpha$.
Let $a$ and $b$ be two nodes that appear in this
order in $\alpha$.
If for some $i\in[k]$, and for some $u,v\in \Vp_i$, $\cliquemembervar{u}$ is set to $1$ by $\alpha$ at node $a$ and
$\cliquemembervar{v}$ is queried at $b$, then $v \neq u$ by regularity and,
by definition of $\Ddist$, the answer
to query $\cliquemembervar{v}$ will
be forced to~$0$, either to avoid violating a
\fun or an \edgeaxiom,
or because $i$ is forgotten at $b$.
\end{proof}
Let us call a pair
$(a,b)$ of nodes of $\bprogpi$ \emph{\useful} if
there exists an index~$i$ such that
$\setones{i}{b}=\emptyset$, $i$ is not forgotten at $b$, and
  the set $\setzeros{i}{b}
  \setminus \setzeros{i}{a}$ is \generic{\rp}{\qp}. In particular, if $a$ appears before $b$ in some path, then $V_i^1(a)=\emptyset$ and $V_i^0(a)\subseteq V_i^0(b)$.
For each \useful pair $(a,b)$,
let $i(a,b)$ be an arbitrary but fixed index witnessing
that $(a,b)$ is \useful.
A path is said to \emph{\frugally} traverse a \useful pair $(a,b)$ if
it goes through $a$ and $b$ in that order and sets at most
$\lceil{k/\tp}\rceil$ variables to $1$ between $a$ and $b$ (with $a$
included and $b$ excluded).

As already mentioned, the proof of \refth{thm:main} is based on a bottleneck counting argument in the spirit of~\cite{Haken85Intractability}, with the twist that we consider pairs of bottleneck nodes.
To establish the theorem we make use of the following two lemmas which will be proven
subsequently.

\begin{lemma}
  \label{lemma:legitimate}
  Every path in the support of $\Ddist$ \frugally traverses a
  \useful pair.
\end{lemma}

\begin{lemma}
  \label{lemma:probability}
  For every \useful pair
  $(a,b)$, the probability that a random $\randomvaralpha$ chosen from
  $\Ddist$ \frugally traverses $(a,b)$ is at most
  $2\spar^{-\newepsilon\rp/2}$.
\end{lemma}

Combining the above lemmas, it is immediate to prove \refth{thm:main}.
By Lemma~\ref{lemma:legitimate} the probability that a random path
$\randomvaralpha$ sampled from $\Ddist$ \frugally traverses some \useful
pair is $1$. By Lemma~\ref{lemma:probability}, for any fixed
\useful pair $(a,b)$, the probability that a random
$\randomvaralpha$ \frugally traverses $(a,b)$ is at most
$2\spar^{-\newepsilon\rp/2}$. By a standard union bound argument,
it follows that the number of \useful pairs is
at least $\frac{1}{2}\spar^{\newepsilon\rp/2}$, so the number of nodes in $\bprogpi$
cannot be smaller than~$\Bigomega{\spar^{\newepsilon\rp/4}} \geq \Bigomega{\spar^{\newepsilon k/\tp^2}}$
(recall that $r\geq 4k/t^2$ according to Definition~\ref{def:propertyP}).

\smallskip
To conclude the proof it remains only to establish
Lemmas~\ref{lemma:legitimate} and~\ref{lemma:probability}.
\begin{proof}[Proof of Lemma~\ref{lemma:legitimate}]
  Consider any
  path in the support of $\Ddist$. As we already remarked, this path ends in the $i^*$th \cliqueaxiom
  for some $i^* \in [k]$ which in particular implies that $V^1_{i^\ast}(b)=\emptyset$ and that $i^\ast$
  is not forgotten at any $b$ along this path.
  By Claim~\ref{claim:atmostk}, the path sets
  at most~$k$ variables to~$1$ and hence we can split it into $t$ pieces by
  nodes $a_0,a_1,\ldots,a_t$ ($a_0$ is the source, $a_t$ the sink) so that
  between~$a_j$ and~$a_{j+1}$ at most $\lceil k/t\rceil$ variables are set to 1.
  It remains to prove that for at least one $j \in [t]$ the set
\begin{equation}
  \label{eq:Wj}
  \Wp_j = \setzeros{i^*}{a_j}
  \setminus \setzeros{i^*}{a_{j-1}}
\end{equation}
is \generic{\rp}{\qp}.
Note that this will prove Lemma~\ref{lemma:legitimate}
since by construction $(a_{j-1},a_j)$ is then a
pair that is \frugally traversed by the path.

Towards contradiction, assume instead that no $\Wp_j$ is \generic{\rp}{\qp},
\ie that for all $j\in [\tp]$ there exists a set of
vertices $\Rp_j\subseteq V$ with
$\setsize{\Rp_j} \leq \rp$ such that
$\Setsize{\commonneighbourhood{\Rp_j}{\Wp_j}} \leq q$.
Let $\Rp = \bigcup_{j \in [\tp]} \Rp_j$.
Since the path ends in the $i^*$th \cliqueaxiom
we have $\setzeros{i^*}{a_\tp} = \Vp_{i^*}$. It follows that the sets $\Wp_1,\ldots,\Wp_\tp$
in~\refeq{eq:Wj} form a partition of~$\Vp_{i^*}$, and therefore
\begin{equation}
  \Setsize{\commonneighbourhood{\Rp}{\Vp_{i^*}}}
  = \sum_{j \in [\tp]} \Setsize{\commonneighbourhood{\Rp}{\Wp_j}}
  \leq \sum_{j \in [\tp]}
  \Setsize{\commonneighbourhood{\Rp_j}{\Wp_j}}
  \leq \tp\qp
  \eqperiod
\end{equation}
Since
$\setsize{\Rp} \leq
\sum_{j \in [\tp]} \setsize{\Rp_j} \leq \tp\rp$
this contradicts the assumption that $\Vp_{i^*}$ is
\generic{\tp\rp}{\tp\qp}.
Lemma~\ref{lemma:legitimate} follows.
\end{proof}

\begin{proof}[Proof of Lemma~\ref{lemma:probability}]
  Fix a \useful pair $(a,b)$. Let $\eventE$ denote the event
  that a random path sampled from~$\Ddist$
  \frugally traverses $(a,b)$. Let
     $i^* = i(a,b)$,
     $\unionsetones = \bigcup_{j \in
       [k]} \setones{j}{a}$, and
     $\Wp = \setzeros{i^*}{b} \setminus \setzeros{i^*}{a}$.
Notice that~$\Wp$ is guaranteed to be  \generic{\rp}{\qp} by our definition of $i(a,b)$.
  Since   $\Gp$ is \cliquedense{k}{\tp}{\rp}{\spar}{\newepsilon} by assumption, this
  implies that $\Wp$ is \robgeneric{\tp\rp}{\rp}{\qp'}{\spar}, and we let
  $\Sp$ be the set that
  witnesses this as per
  \refdef{def:robgeneric}.
 We bound the probability of the event
  $\eventE$ by a case analysis based on the size of the
  set~$\unionsetones$.
  We remark that
  all probabilities in the calculations that follow are over
  the choice
  of~$\sample{\randomvaralpha}{\Ddist}$.

 {\textbf{Case 1 ($\setsize{\unionsetones} > \rp/2$)}}: In this case, we simply prove
 that already the probability of reaching $a$ is small.
By definition of $\unionsetones$, we have that $\setsize{\betaone(a)} = \setsize{\unionsetones}$.
Recall that every answer $1$ is necessarily the result of a $\spar^{-(1+\newepsilon)}$-biased
coin flip, and that all these decisions are irreversible. That is, if a path
ever decides to set a variable in $V^1(a)$ to 0, then its case is lost and it is guaranteed to miss $a$. Thus we
can upper bound
the probability of the event $\eventE$ by the probability
 that a random $\randomvaralpha$ passes through $a$, and, in particular,
 by the probability of setting all variables in $\betaone(a)$ to~$1$
  as follows:
 \begin{align}
 \prob{\eventE }
 & \leq
 \prob{ \randomvaralpha \text{ passes through } a}
 \\&
 \leq  \big(\spar^{-(1+\newepsilon)}\big)^{\setsize{\betaone(a)}}
 \\&
  \leq   \spar^{-\newepsilon\setsize{\unionsetones}}
  \\&
  \leq  2\spar^{-\newepsilon\rp/2}
\eqperiod
  \label{eqn:oneside}
\end{align}

{\textbf{Case 2 ($\setsize{\unionsetones} \leq \rp/2$)}}:
For every path $\alpha$, let $\Rp(\alpha)$ denote the set of
vertices $u$ set to $1$ by the path~$\alpha$
at some node between
$a$ and $b$ (with $a$ included and $b$ excluded);
note that $\Rp(\alpha) = \emptyset$ if
$\alpha$ does not go through $a$ and $b$,
and that
$\setsize{\Rp(\alpha)}\leq \lceil{k/\tp}\rceil$ for all paths $\alpha$ that satisfy the
event~$\eventE$.
For the sets
\begin{subequations}
\begin{align}
  \mathcal{\Rp}_0
  &= \set{ \Rp : \setsize{\Rp} \leq \lceil{k/\tp}\rceil \text{ and }
    \Setsize{\commonneighbourhood{\Rp \cup \unionsetones}{\Wp}} < \qp'}
  \\
  \mathcal{\Rp}_1
  &= \set{ \Rp : \setsize{\Rp} \leq \lceil{k/\tp}\rceil \text{ and }
    \Setsize{\commonneighbourhood{\Rp \cup \unionsetones}{\Wp}} \geq \qp'}
\end{align}
\end{subequations}
we have that
\begin{equation}
  \prob{\eventE } = \prob{\eventE \text{ and }
  \randomvarR \in \mathcal{\Rp}_0 } +
  \prob{\eventE \text{ and } \randomvarR \in \mathcal{\Rp}_1 }\eqperiod
  \label{eqn:split}
\end{equation}

The first term in \eqref{eqn:split} is bounded from above by the
probability of
  $\randomvarR \in \mathcal{\Rp}_0$.
  Note that $\setsize{\Rp} \leq
  \lceil{k/\tp}\rceil \leq 2k/\tp \leq \tp\rp/2$ (since $\rp\geq 4k/\tp^2$) for $\Rp\in\mathcal{\Rp}_0$.
  Hence we have $|R\cup V^1(a)|\leq \tp\rp/2+\rp/2\leq \tp\rp$ and
  therefore
  $\setsize{(\Rp \cup \unionsetones) \cap \Sp} \geq \rp$ by
  the choice of~$\Sp$. Thus, the probability of $\randomvarR \in
  \mathcal{\Rp}_0$ is bounded by the probability that $\setsize{\randomvarR \cap \Sp}
  \geq \rp/2$ since $\setsize{\unionsetones} \leq
  \rp/2$. But since $S$ is small, we can now apply the union bound and conclude that
  \begin{align}
    \prob{\eventE \text{ and } \randomvarR \in \mathcal{\Rp}_0 }
    & \leq
    \prob{ \randomvarR \in \mathcal{\Rp}_0 }
    \\ & \leq
    \prob{ \setsize{\randomvarR \cap \Sp} \geq \rp/2 }
    \\ &
    \leq \binom{\setsize{\Sp}}{\rp/2 }(\spar^{-(1+\newepsilon)})^{\rp/2 } \label{eqn:irreversibility}
    \\ &
    \leq
    \setsize{\Sp}^{\rp/2 }\spar^{-(1+\newepsilon)\rp/2}
    \\ &
    \leq
    \spar^{-\newepsilon\rp/2} \eqcomma \label{eqn:theother}
  \end{align}
  where for \eqref{eqn:irreversibility} we used the same ``irreversibility'' argument as in Case 1.

  We now bound the second term in \eqref{eqn:split}. First note that, by definition of~$W\!$, if~$\alpha$ is a path that passes through~$a$ and~$b$ in this order, then
all $u \in \Wp$ must be set to~$0$ in~$\alpha$ at
  some node between $a$ and $b$. For each path in the support of $\Ddist$
  that passes through~$a$ and~$b$,
  some of the vertices in $\Wp$ will be set to
  zero as a result of a coin flip and others will be forced choices.

     Fix a path $\alpha$ contributing to the second term in \eqref{eqn:split}. We claim that along
     this path all the $\geq q'$ variables in $\commonneighbourhood{\Rp(\alpha) \cup \unionsetones}{\Wp}$ are set to $0$ as a result of a coin flip.
     Indeed, since $\setones{i^*}{b}=\emptyset$ and $i^*$ is not forgotten at $b$, by the monotonicity
     property the same holds for every node along $\alpha$ before~$b$. This implies that the answer to
     a query of the form $\cliquemembervar{u}\ (u\in W)$ made along $\alpha$ cannot be forced by neither item~\ref{item:forced-choice1}
     (forgetfulness) in the definition of $\Ddist$ nor by a functionality axiom. Moreover, since $V^1(c)\subseteq \Rp(\alpha) \cup \unionsetones$
     for any node $c$ on the path $\alpha$ between $a$ and~$b$, it holds that all variables $\cliquemembervar{u}$
     with $u\in\commonneighbourhood{\Rp(\alpha) \cup \unionsetones}{\Wp}$ can not be forced to $0$ by an edge axiom either.

     The analysis of the second term in \eqref{eqn:split} is completed by the same type of argument
     as in Case 1, where we again use the fact that, due to
     the read-once property of the branching program, the decisions that the
     random path makes are irreversible:
       \begin{align}
    \prob{ \eventE \text{ and } \randomvarR \in \mathcal{\Rp}_1 }
    &\leq
    \prob{\randomvaralpha\ \text{flips}\ \geq \qp'\ \text{coins}\ \text{and gets 0-answers}}
    \\
    &\leq (1-\spar^{-(1+\newepsilon)})^{\qp'}
    \\
    & \leq \spar^{-\newepsilon\rp/2} \eqperiod \label{eqn:theother2}
  \end{align}

  Adding~\eqref{eqn:theother} and~\eqref{eqn:theother2} we obtain the lemma.
\end{proof}
\section{Random Graphs Are Almost Surely Clique-Dense}
\label{sec:lowerbound2}

In this section we show that \aas an \ErdosRenyi
random graph $\sample{\randomvarG}{\randG{\np}{p}}$ is
\cliquedense{k}{\tp}{\rp}{\spar}{\newepsilon} for the right choice of
parameters.

\begin{theorem}
	\label{thm:erdos-renyi-clique-dense}
	For any real constant $\oldxi \in (0,1/4)$,
        any sufficiently large integer~$n$,
        any positive integer~$k\leq {n^{1/4-\oldxi}}$,
        and any real $\oldeta>1$,
        if $\sample{\randomvarG}{\randG{\np}{\np^{-2\oldeta/(k-1)}}}$ is an
        \ErdosRenyi random graph then
        with probability at least $1-\exp(-\sqrt{\np})$ it holds that
        $\randomvarG$ is
         \cliquedense{k}{\tp}{\rp}{\spar}{\newepsilon} with
	$\tp={32\oldeta}/{\oldxi}$ and $\spar= \sqrt{\np}$.
\end{theorem}

As a corollary of \refthm{thm:main} and \refthm{thm:erdos-renyi-clique-dense} we
obtain Theorem \ref{thm:k-clique-erdos-renyi-block},
the main result of
this paper.

\begin{proof}[Proof of Theorem \ref{thm:k-clique-erdos-renyi-block}]
Clearly $\tp \geq 1$ as required by Definition~\ref{def:propertyP}.
We can also assume w.l.o.g. that $\tp \leq k$ since otherwise $k/\xi^2\leq
32/(\xi\epsilon)\leq O(1)$ and the bound becomes trivial.
By plugging in the parameters given by  \refthm{thm:erdos-renyi-clique-dense}
to \refthm{thm:main} we immediately get that 
any regular refutation $\pi$ of $\cliqueblock{\randomvarG}{k}$
has \lengthsize
		\begin{align}
	\label{eq:clique-G(n,p)-actual-lb}
	\setsize{\pi}
	\geq
	\Bigomega{\spar^{\newepsilon k/\tp^2}}
    \geq
    n^{\bigomega{k/\xi^2}}
	\eqcomma
	\end{align}	
as stated.
\end{proof}

We will spend the rest of this section proving \refthm{thm:erdos-renyi-clique-dense}.

Let $\delta = 2\oldeta/(k-1)$.
We show that, with probability
        at least $1-\E^{-\sqrt{\np}}$, the random graph
        $\randomvarG$ is
        \cliquedense{k}{\tp}{\rp}{\spar}{\newepsilon} for parameters as in the statement of the
        theorem, $\rp=4k/\tp^2$ and $\qp=\frac{\np^{1-\tp\delta\rp}}{4k\tp}$.

Recall that $\qp'=\newepsilon\rp\spar^{1+\newepsilon}\log\spar$.
Let us argue that the parameters we use satisfy constraints
	\begin{align}
	\label{eq:bound-tdeltar}
	& \tp\delta \rp
	\leq
	\frac{\oldxi}{2}
	\eqcomma
	\\
	\label{eq:bound-tr}
	& \log k + \tp\rp\log \np
	\leq
	\frac{\np^{1-\tp\delta\rp}}{32k}\cdot\frac{2\log\np}{\np^{1/2}}
	\eqcomma
		\\
	\label{eq:bound-x-lb}
	& \frac{\qp\np^{-\tp\delta\rp}\spar}{16\tp\rp}
	\geq
	\frac{\np^{1+\oldxi}}{256}
	\eqcomma
	\\
	\label{eq:bound-repsilon}
	& \qp'
	\leq \frac{\qp\np^{-\tp\delta \rp}}{4} \cdot \frac{\log \np}{\np^{\oldxi/2}}
	\eqcomma
	\\
	\label{eq:tr}
	& \tp\rp
	\leq \frac{\qp}{2}
	\eqcomma
	\end{align}
which will be used further on in the proof.	

As a first step note that
	\begin{equation}
	\tp\delta\rp
	=
	\frac{8\oldeta k}{t(k-1)}
	\leq \frac{\oldxi}{2}	
	\eqcomma
	\end{equation}
	and hence \refeq{eq:bound-tdeltar} holds. Equation
        \refeq{eq:bound-tr} follows from the
        chain of inequalities
	\begin{align}
	\log k + \tp\rp\log \np \leq 2\tp\rp\log \np
	=
	\frac{8k\log \np }{\tp}
	\leq
	\frac{k\log \np }{16}
	\leq
	\frac{\np^{1/2-2\oldxi}\log \np}{16k}
	\leq
	\frac{\np^{1-\tp\delta\rp}}{32k}\cdot\frac{2\log\np}{\np^{1/2}}
	\eqperiod
	\end{align}
	To obtain \refeq{eq:bound-x-lb} observe that
	\begin{equation}
	\frac{\qp\np^{-\tp\delta\rp}\spar}{16\tp\rp}
	=
	\frac{\np^{1-2\tp\delta\rp + 1/2}}{256k^2}
	\geq
	\frac{\np^{1-2\tp\delta\rp+2\oldxi}}{256}
	\geq
	\frac{\np^{1+\oldxi}}{256}
	\eqperiod
	\end{equation}
To see that \refeq{eq:bound-repsilon} holds, note that
	\begin{equation}
	\qp'
	=
	\frac{2\oldxi k\np^{{(1+\newepsilon)}/{2}}\log \np}{\tp^2}
	\leq
	\frac{k^2\np^{{(1+\oldxi)}/{2}}\log \np}{16k\tp} \label{eq:def-epsilon-appl}
	\leq
	\frac{\np^{1-{3\oldxi/2}}\log \np}{16k\tp}
	\leq
	\frac{\qp\np^{-\tp\delta \rp}}{4} \cdot \frac{\log \np}{\np^{\oldxi/2}}
	\eqperiod
	\end{equation}
	Finally, for \eqref{eq:tr}, we just observe that
	\begin{align}
		\label{eq:tp-proof}
	\tp\rp
	=
	\frac{4k}{\tp}
	\leq
	\frac{k^3}{8k^2}
	\leq
	\frac{n^{1-\tp\delta\rp}}{8k\tp}
	=
	\frac{\qp}{2}
	\eqcomma
	\end{align}
using the fact that $k \geq \tp$ and $k^3 \leq n^{1-\tp\delta\rp}$.

\medskip

	We must now prove that \aas $\randomvarG$ is \cliquedense{k}{\tp}{\rp}{\spar}{\newepsilon} for the  chosen parameters.
	All probabilities in this section are over the choice of~$\randomvarG$, and all previously introduced concepts like $\commonneighbourhood{\Rp}{\Wp}$, neighbour-denseness etc. should be
understood with respect to~$\randomvarG$ as well (so that they are actually random variables and events in this sample space).
	Let $\Vp = \Vp(\randomvarG)$
	and $\Vp_1 \cup \cdots \cup \Vp_k = \Vp$ be a balanced $k$-partition of $\Vp$.
	
	The fact that \aas $\Vp_i$ is \generic{\tp\rp}{\tp\qp}  for all $i\in[k]$ is
        quite immediate. First, for any $i\in[k]$ and any $\Rp \subseteq \Vp$ with $\setsize{\Rp}\leq \tp\rp$,	
        \begin{align}\Expect{\Neighcsize{\Rp}{\Vp_i}}
        =
        \setsize{\Vp_i\setminus \Rp}\np^{-\delta \setsize{\Rp}}
        \geq
        \left(\frac{\np}{k}-\tp\rp\right)
        \np^{-\delta\tp \rp}
        \geq
        \left(\frac{\np}{k}-\frac{\qp}{2}\right)
        \np^{-\delta\tp \rp}
        \geq
        \label{eq:expect-V_i-last}
        \frac{\np^{1-\delta\tp \rp}}{2k}
        \eqcomma
        \end{align}
        where 
        the second-to-last inequality follows from \eqref{eq:tr} and 
        the last inequality from the trivial fact that $\qp\leq \frac{\np}{k}$.
         Hence, we can bound the probability that there exists an $i\in[k]$ such that $\Vp_i$ is not \generic{\tp\rp}{\tp\qp} by
    \begin{align}
    & \PROB{\exists i\in [k]\ \exists \Rp\subseteq \Vp,\ \setsize{\Rp}= \lfloor\tp\rp\rfloor \land
    	\Neighcsize{\Rp}{\Vp_i}\leq \tp\qp} \\
	& \;\;\;\;\; \leq
       \label{eq:use-bound-union}
	k \binom{\np}{\tp\rp}\max_{i,\Rp}\PROB{\Neighcsize{\Rp}{\Vp_i}\leq \tp\qp}
	\\
	\label{eq:use-bound-tb}
	&  \;\;\;\;\; \leq
	k \np^{\tp\rp}\max_{i,\Rp}\PROB{\Neighcsize{\Rp}{\Vp_i}\leq \frac{\np^{1-\tp\delta\rp}}{4k}}
	\\	
	\label{eq:use-bound-chernoff}
	&  \;\;\;\;\; \leq
	k \np^{\tp\rp}\exp\left(-\frac{\np^{1-\tp\delta\rp}}{16k}\right)
	\\
	\label{eq:use-bound-tr}
	&  \;\;\;\;\; \leq
	\exp\left(-\frac{\np^{1-\tp\delta\rp}}{32k} \cdot\left(2-2\frac{\log \np}{\np^{1/2}} \right)\right)
	\\
	\label{eq:use-bound-tdeltar}
	&  \;\;\;\;\; \leq
	\E^{-\sqrt{\np}}
	\eqperiod
	\end{align}
We note that \refeq{eq:use-bound-union} is a
union bound, \refeq{eq:use-bound-tb} follows from
        the definition of $\qp$,  \refeq{eq:use-bound-chernoff} is the
        multiplicative form of Chernoff bound (note that the events $v\in\commonneighbourhood{\Rp}{\Vp_i} (v\in \Vp\setminus \Rp)$
        are mutually independent),
        \refeq{eq:use-bound-tr} follows from
        \refeq{eq:bound-tr}, and \refeq{eq:use-bound-tdeltar} holds for large enough $\np$ by
        \refeq{eq:bound-tdeltar} and the fact that $\oldxi<1/4$ and $k<n^{1/4}$.
	
	All that is left to prove is that \aas $\randomvarG$ satisfies property \ref{item:hyp-awesome} in Definition~\ref{def:propertyP},	that is that
	every \generic{\rp}{\qp} set $\Wp\subseteq \Vp$
	is \robgeneric{\tp\rp }{\rp}{\qp'}{\spar}.
	For shortness let $\eventA$ be the event that $\randomvarG$ satisfies this
	property.
	We wish to show that $\prob{\lnot \eventA}\leq \E^{-\bigomega{\np}}$, and it turns out that due to our choice of parameters we can afford to use the crude union bound over all $2^n$ choices of $W$.
	
	To be more specific, let $Q(W)$ denote the event that $W$ is  \generic{\rp}{\qp}. Given an \generic{\rp}{\qp} set $\Wp\subseteq\Vp$ we will define
	a set $\Sp_\Wp$ which will be a ``candidate witness'' of the fact that $\Wp$ is \robgeneric{\tp\rp}{\rp}{\qp'}{\spar}.
	First observe that, since $\Wp$ is \generic{\rp}{\qp} and $\qp' \leq \qp$ by~\eqref{eq:bound-repsilon}, any set $\Rp\subseteq \Vp$ with $\setsize{\Rp}\leq \tp\rp$ and $\Neighcsize{\Rp}{\Wp}\leq \qp'$ must be such that $\setsize{\Rp} > \rp$. We will use a sequence of such sets $\Rp$ and construct $\Sp_\Wp$ in a greedy fashion. To this end, the following definition will be useful.
	A tuple of sets $(\Rp_1,\ldots, \Rp_m)$ is said to be \emph{\disj{\rp}} if $\Setsize{\Rp_i \cap \big(\bigcup_{j<i}\Rp_j\big)}\leq \rp$ for every~$i\in[m]$.

	Fix an arbitrary ordering of the subsets of $\Vp$. Define $\vec \Rp_\Wp=(\Rp_1,\dots,\Rp_m)$ to be
	a maximally long tuple such that, for every $i= 1, \ldots, m$, the set $\Rp_i$ is the first in the ordering such that $\setsize{\Rp_i}\leq \tp\rp$, $\Neighcsize{\Rp_i}{\Wp}\leq \qp'$ and $\Setsize{\Rp_i \cap \big(\bigcup_{j<i}\Rp_j\big)}\leq \rp$. Note that $\vec\Rp_\Wp$ is \disj{\rp}.
	Now let $\Sp_\Wp=\bigcup_{i\leq m} \Rp_i$.
	
Observe that, by maximality of $\vec\Rp_\Wp$, any set $\Rp\subseteq \Vp$ with $\setsize{\Rp}\leq \tp\rp$ and $\Neighcsize{\Rp}{\Wp}\leq \qp'$ must be such that $\setsize{\Rp\cap \Sp}> \rp$.
	This implies that if $\setsize{\Sp_\Wp} \leq \spar$ then $\Sp_\Wp$ witnesses the fact that $\Wp$ is \robgeneric{\tp\rp }{\rp}{\qp'}{\spar}. Therefore we have that
	\begin{equation}
	\prob{\lnot \eventA}\leq \prob{\exists \Wp\subseteq \Vp,\ Q(W)  \land \setsize{\Sp_\Wp}>\spar}\eqperiod
	\end{equation}
	
	Moreover, let $\setWR$ be the collection of all pairs $(\Wp,\vec \Rp)$ such that \mbox{$\Wp\subseteq \Vp$}, %
	$\vec \Rp=(\Rp_1,\ldots,\Rp_\yp)$ for $\yp= \ceiling{\spar/\tp\rp}$,  $\Rp_j \subseteq \Vp$ and \mbox{$0<\setsize{\Rp_j}\leq \tp\rp$} for each $j\in [\yp]$,
	and  $\vec \Rp$ is \disj{\rp}.
	Notice that if there exists an \generic{\rp}{\qp} $\Wp$ such that  $\vec \Rp_\Wp=(\Rp_1,\ldots,\Rp_m)$ and $\setsize{\Sp_\Wp}>\spar$, then
	$m\geq \yp$ and {$(\Wp,(\Rp_1,\ldots,\Rp_\yp))\in \setWR$}.
	Furthermore, by definition of $\vec \Rp_\Wp$,
	for every~$j \in [\yp]$
	it holds that $\Neighcsize{\Rp_j}{\Wp}\leq \qp'$.
	Hence we can conclude that
	\begin{align}
	\prob{\lnot \eventA}
	&\leq
	\Prob{\exists (\Wp,\vec \Rp)\in \setWR,\ Q(W) \land \forall j \in [\yp],\ \Neighcsize{\Rp_j}{\Wp}\leq \qp'}
	\\
	& \leq
	2^\np\np^{\tp\rp\yp}
	\max_{(\Wp,\vec R)\in \setWR}
	\Prob{ Q(W) \land \forall j \in [\yp],\ \Neighcsize{\Rp_j}{\Wp}\leq \qp'}
	\\
	\label{eq:use-bound-repsilon}
	& \leq
	2^\np\np^{\spar}
	\max_{(\Wp,\vec \Rp)\in \setWR}
	\Prob{Q(W) \land\forall j \in [\yp],\ \Neighcsize{\Rp_j}{\Wp}\leq\frac{\qp}{4}\np^{-\tp\delta \rp}}
	\eqcomma
	\end{align}
	where \refeq{eq:use-bound-repsilon} follows for $\np$ large enough from the bound in \refeq{eq:bound-repsilon}.

Now fix $(\Wp,\vec \Rp)\in \setWR$ and
let~$\Rp_j^d$ (resp. $\Rp_j^c$) be the subset of $\Rp_j$ disjoint from (resp. contained in) $\bigcup_{j^\prime <j} \Rp_{j^\prime}$.
Since $\setsize{\Rp_j^c}\leq \rp$ by definition, it holds that if $\Wp$ is \generic{\rp}{\qp} then $\Neighcsize{\Rp_j^c}{\Wp}> \qp$.
Let $\neighcineq{j}$ be the event that  $\Neighcsize{\Rp_j^c}{\Wp}> \qp$ and
$\Neighcsize{\Rp_j}{\Wp}\leq \frac{\qp}{4}\np^{-\tp\delta
	\rp}$. Note that
$\Prob{Q(W) \land\forall j \in [\yp],\ \Neighcsize{\Rp_j}{\Wp}\leq\frac{\qp}{4}\np^{-\tp\delta \rp}}$
is at most
$\Prob{\forall j \in [\yp],\ \neighcineq{j}}$.
Let $\neighcineqp{j}$ be the event that $\neighcineq{j'}$ holds for all
$j' \in [j-1]$. We have that
	\begin{equation}
	\label{eq:product}
	\Prob{\forall j \in [\yp],\ \neighcineq{j}}
	=
	\prod_{j \in [\yp]}\Condprob{\neighcineq{j}\ }{\ \neighcineqp{j} }
	\eqperiod
	\end{equation}
We can consider the factors of the previous product separately and bound each one by
\begin{align}
& \nonumber \Condprob{\neighcineq{j}\ }{\ \neighcineqp{j} }
        \\
&	 \leq
	\sum_{\substack{U\subseteq \Wp\\ \setsize{U}\geq \qp}}
	\CONDPROB{\Neighcsize{\Rp^d_j}{U}\leq\frac{\qp}{4}\np^{-\tp\delta \rp}\ }{\ \neighc{\Rp^c_j}{\Wp}=U \land \neighcineqp{j}} \cdot \CONDPROB{\neighc{\Rp^c_j}{\Wp}=U \ }{\ \neighcineqp{j}}
	\\
&	\leq
	\label{eq:indip}
	\sum_{\substack{U\subseteq \Wp\\ \setsize{U}\geq \qp}}
	\PROB{\Neighcsize{\Rp^d_j}{U}\leq\frac{\qp}{4}\np^{-\tp\delta \rp}} \cdot \CONDPROB{\neighc{\Rp^c_j}{\Wp}=U \ }{\  \neighcineqp{j}}
	\\
&	\leq
	\label{eq:R_j-chernoff}
	\sum_{\substack{U\subseteq \Wp\\ \setsize{U}\geq \qp}}
	\exp\left(-\frac{	\qp\np^{-\tp\delta \rp}}{16}\right) \cdot
	\CONDPROB{\neighc{\Rp^c_j}{\Wp}=U \ }{\ \neighcineqp{j}}
	\\
&	=
	 \exp\left(-\frac{	\qp\np^{-\tp\delta \rp}}{16}\right) \cdot
	 \sum_{\substack{U\subseteq \Wp\\ \setsize{U}\geq \qp}}
	 \CONDPROB{\neighc{\Rp^c_j}{\Wp}=U \ }{\ \neighcineqp{j}}
	\\
&	\leq
	 \exp\left(-\frac{\qp\np^{-\tp\delta \rp}}{16}\right)
	\eqperiod
	\end{align}
	Equation \refeq{eq:indip} follows from the independence of any two events
        that involve disjoint sets of potential edges and \refeq{eq:R_j-chernoff} follows from the multiplicative Chernoff bound and the fact that
	\begin{align}
	\Expect{\Neighcsize{\Rp_j^d}{U}}
	=
	\setsize{U\setminus \Rp_j^d}\np^{-\delta \setsize{\Rp_j^d}}
	\geq (\setsize{U}-\tp\rp)\np^{-\delta\tp \rp}
	\geq
	\frac{\qp}{2}\np^{-\delta\tp \rp}
	\eqperiod
	\end{align}
	So, putting everything together, we have that
	\begin{equation}
	\prob{\lnot \eventA}
	\leq
	2^\np\np^{\spar} \exp\left(-\frac{	\qp\np^{-\tp\delta \rp}\yp}{16}\right)
	\leq \E^{(\log 2)\np+\sqrt{\np}\log \np -(\np^{1+\oldxi})/256}
	\leq \E^{-\bigomega{\np}}
	\eqcomma
	\end{equation}
	where the last inequality holds for $\np$ large enough, and the second to last inequality follows immediately from the bound in \refeq{eq:bound-x-lb}.
	This concludes the proof of \refthm{thm:erdos-renyi-clique-dense}.
\section{State-of-the-Art Algorithms for Clique}
\label{sec:algorithms}

\newcommand{\expand}{\mathtt{expand\xspace}}
\newcommand{\colourOrder}{\mathtt{colourOrder\xspace}}
\newcommand{\permute}{\mathtt{permute\xspace}}
\newcommand{\glob}{\textbf{global\xspace}}

\newcommand{\inc}{\mathit{incumbent\xspace}}
\newcommand{\bounds}{\mathit{bounds\xspace}}
\newcommand{\order}{\mathit{order\xspace}}
\newcommand{\sol}{\mathit{solution\xspace}}
\newcommand{\found}{\mathit{found\xspace}}

\newcommand{\True}{\textsf{true\xspace}}
\newcommand{\False}{\textsf{false\xspace}}
\newcommand{\Then}{\textbf{then\xspace}}
\newcommand{\Myelse}{\textbf{else\xspace}}
\newcommand{\Myif}[2]{\textbf{if\xspace}{} {#1} \Then{} {#2}}
\newcommand{\Downto}{\textbf{down to\xspace}}
\newcommand{\Or}{\textbf{or\xspace}}

In this section we describe state-of-the-art algorithms
for maximum clique and explain how regular resolution
proofs bound from below the running time of
these algorithms.

At the heart of most (if not all) of the state-of-the-art algorithms 
for maximum clique 
is a backtracking search, which 
in its simplest form examines
all maximal cliques by enlarging a set of vertices that
form a clique and backtracking when it certifies
that the current set forms
a maximal clique. 
A classical 
example of such a backtracking search is 
the Bron–Kerbosch~\cite{BK73FindingAllCliques} algorithm
which enumerates
all maximal cliques in a graph. 
This algorithm can be adapted to find a maximum
clique as done in~\cite{CP90ExactAlgorithm} 
improving the running time considerably by
using a branch and bound strategy.
At some point in the search tree it becomes clear that 
the current search-branch will not lead to a clique larger
than the largest one found so far---in such cases the 
algorithm cuts off the search and backtracks immediately.

The most successful algorithms in practice are
search trees with clever branch and bound strategies.
In this section we will discuss 
the algorithm by Östergård~\cite{Ostergard02FastAlgorithm} using
Russian doll search 
and a collection of algorithms that use
colour-based branch and bound strategies~\cite{
Wood97AlgorithmForFinding,
Fahle02SimpleAndFast,
TS03EfficientBranchAndBound,
TK07EfficientBranchAndBound,
KJ07ImprovedBranchAndBound,
TSHTW10SimpleAndFaster,
ST10NewImplicitBranching,
SRJ11ExactBitParallel,
SMRH13ImprovedBitParallel,
SLB14InitialSorting,
SLBNP16ImprovedInitialVertex,
TYHNIW16MuchFaster}.

\paragraph{Östergård's algorithm}

Östergård's algorithm~\cite{Ostergard02FastAlgorithm} is a branch and
bound algorithm that uses Russian doll search as a pruning strategy:
it considers smaller subinstances recursively and solves them in
ascending order using previous solutions as upper bounds.
This algorithm, which is the main component of the Cliquer software,
is often used in practice and has been available online since
2003~\cite{NO03Cliquer}. Cliquer is also the software
of choice to compute maximum cliques in the open source mathematical
software SageMath~\cite{Sage2017}.

\begin{algorithm}[t]
\begin{algorithm2e}[H]
$\mathtt{Cliquer}(G)$:\;
\Begin{
$G \leftarrow \permute(G)$\;
$\inc \leftarrow \emptyset$\;
\For{$i =n$ \Downto{} $1$}{ \label{line:for1}
$\found \leftarrow \False$\;
$\expand (G[V_i\cap \neigh{v_i}], \set{v_i})$\;
$\bounds[i] \leftarrow \setsize{\inc}$ \label{line:bounds1}
}
\Return $\inc$ \;
}
$\expand(H,\sol)$:\;
\Begin{
\While{$V(H) \neq \emptyset$}{
\Myif{$\setsize{\sol} + \setsize{V(H)}\leq \setsize{\inc}$}{\Return}\; \label{line:boundsize1}
$i\leftarrow \min\set{j\mid v_j \in V(H)}$\; 
\Myif{$\setsize{\sol} + \bounds[i] \leq \setsize{\inc}$}{\Return}\; \label{line:bound1}
$\sol' \leftarrow \sol \cup \set{v_i}$\; \label{line:sol1}
$V' \leftarrow V(H)\cap \neigh{v_i}$\;
{$\expand(H[V'],\sol')$}\; \label{line:recurse1}
\Myif{$\found = \True$}{\Return}\; 
$H \leftarrow H\setminus \set{v_i}$\; \label{line:notv1}
}
\If{$\setsize{\sol'} > \setsize{\inc}$}{
$\inc \leftarrow \sol'$\; \label{line:newmax1}
$\found \leftarrow \True$ \label{line:found1}
} 
\Return\;
}
\end{algorithm2e}
\caption{$\mathtt{Cliquer}(G)$ algorithm}
\label{alg:Cliquer}
\end{algorithm}

The $\mathtt{Cliquer}(G)$ algorithm described in
Figure~\ref{alg:Cliquer} is essentially the same as Algorithm 2 in~\cite{Ostergard02FastAlgorithm}.
The algorithm first permutes the vertices of $G$ according to some criteria.
Let $v_{1}, \ldots, v_{n}$ be the enumeration of $V(G)$ induced by said
permutation, and $V_i = \set{v_i, \ldots,v_n}$ for $i\in[n]$.
In practice this permutation has a large impact on the running time of
the algorithm, but for our analysis the knowledge of the specific
order is irrelevant.

In the main loop (lines~\ref{line:for1}--\ref{line:bounds1}) subgraphs of~$G$ are considered and at each iteration
the size of a maximum clique
containing only vertices of $V_i$ is stored in $\bounds[i]$.
The algorithm keeps the best solution (largest clique) found so far in the global variable $\inc$ which is initially empty. The array $\bounds$ and the flag $\found$ are global variables. The current growing clique is stored in $\sol$ and passed as an argument of the subroutine $\expand$ together with the current subgraph $H\subseteq G$ being considered.

The main subroutine $\expand$ recursively goes through all vertices of $H$ from smallest to largest index. First note that if the size of the current growing clique plus $\setsize{H}$ is not larger than the current maximum clique (line~\ref{line:boundsize1}) then this branch can be cut. 
Moreover, if $v_i$ is the smallest-index vertex in $H$ then $V(H)\subseteq V_i$ and $\bounds[i]$ is an upper bound on the size of a maximum clique in $H$.
This implies that this branch can be cut if the size of the current growing clique plus $\bounds[i]$ is not larger than the current maximum clique (line~\ref{line:bound1}). If it is larger,
the algorithm branches on the vertex $v_i$. 

First $v_i$ is taken to be part of the solution: it is added to (a copy of) the current growing solution, (a copy of) the graph is updated to contain only neighbours of $v_i$ and a recursive call is made (lines~\ref{line:sol1}--\ref{line:recurse1}). If the recursive call finds a clique larger than the current largest clique, it sets the flag $\found$ to true. This allows the algorithm can return to the main routine (line~\ref{line:bounds1}) since a maximum clique containing only vertices of $V_i$ can be at most one unit larger than a maximum clique containing only vertices of $V_{i+1}$. If no larger clique was found, the algorithm then proceeds to the opposite branch choice, that is, taking vertex $v_i$ to not be in the solution (line~\ref{line:notv1}) and considering the next vertex in the ordering. If $V(H)$ is empty and a larger clique has been found, the best solution so far is updated 
and the flag $\found$ is set to true (lines~\ref{line:newmax1}--\ref{line:found1}).

We now argue that the running time of the $\mathtt{Cliquer}(G)$ algorithm is 
bounded from below by the size of a 
regular resolution refutation of 
$\cliqueblock{G}{k}$
up to a constant factor. 
First note that a straightforward modification
of the $\mathtt{Cliquer}(G)$ algorithm gives
an algorithm 
that determines whether $G$ contains a block-respecting $k$-clique.

Given a graph $G$ that does not contain a block-respecting $k$-clique,
the last call of the subroutine $\expand$ in the main loop
(lines~\ref{line:for1}--\ref{line:bounds1}, when $i$ is set to $1$)
can be represented by an ordered decision tree with labelled leafs.
A decision tree is said to be ordered if there exists a linear ordering of the variables 
such that if $x$ is queried before $y$ then $x \prec y$. 
In our setting, the order is determined by the permutation of the vertices, 
and without loss of
generality we assume $v_i \prec v_j$ if $i<j$. 
For each leaf, if $R$ is the set of vertices identified as clique members
by the branch leading to this leaf, then the leaf is
labelled either by a pair $(u,v)$ such that $u,v\in R$ and there is no edge between $u$ and $v$
or by an index $\ell\in[k]$ such that 
all vertices in the $\ell$th block are outside the clique, 
or  by a vertex $v_i$ such that $i=\min\set{j\mid v_j \in \neigh{R}}$ and the largest clique containing only vertices of $V_{i}$
has size at most $k-\setsize{R}-1$. 
For each vertex $v_i$ that labels some leaf, we construct the decision tree
corresponding to the $i$th call of the subroutine $\expand$.

In order to weave these decision trees into a read-once branching program, 
at each leaf labelled $v_i$ we query all non yet queried vertices $v_j$ such that $j<i$ and 
$v_j$ is in the same block as $v_i$. Let $B_i$ denote the set of vertices. 
Observe that taking any vertex in $B_i$ to be in
the clique yields an immediate contradiction since $B_i \cap \neigh{R} = \emptyset$ 
by definition of $i$. 
Moreover note that the branch leading to 
the leaf where all of $B_i$ is taken to be outside the clique does not
contain any query to vertices in $V_{i}$.
We can therefore identify this leaf
with the root of the decision tree corresponding to $v_i$ and
still maintain regularity.
After repeating this procedure at every leaf labelled by some vertex, 
only leafs labelled by indices $\ell\in[k]$ and by pairs $(u,v)$ remain,
which have a direct correspondence to falsified clauses of $\cliqueblock{G}{k}$.
Therefore, the directed graph obtained by this process corresponds to 
a read-once branching program that solves the falsified clause search 
problem on $\cliqueblock{G}{k}$ and 
the bound on the running time follows immediately.

\paragraph{Colour-based branch and bound algorithms}
We consider a class of algorithms
which are arguably the most successful in practice. An extended
survey together with a computational analysis of  
algorithms published until 2012 can be found in~\cite{Prosser12Exact} 
and an overview of algorithms reported since then in~\cite{McCreesh17Thesis}.
These algorithms are
branch and bound algorithms 
that use colouring as a
bounding---and often also as a branching---strategy.
The basic idea is that if a graph can be coloured with $\ell$ colours then
it does not contain a clique larger than $\ell$.

\begin{algorithm}[t]
\begin{algorithm2e}[H]
$\mathtt{Max Clique}(G)$:\;
\Begin{
\glob {} $\inc \leftarrow \emptyset$\;
$\expand (G, \emptyset)$\;
\Return $\inc$ \;
}
$\expand(H,\sol)$:\;
\Begin{
$(\order,\bounds) \leftarrow \colourOrder(H)$\; \label{line:colourOrder}
\While{$V(H) \neq \emptyset$}{
$i\leftarrow \setsize{V(H)}$\; 
\Myif{$\setsize{\sol} + \bounds[i] \leq \setsize{\inc}$}{\Return}\; \label{line:bound}
$v\leftarrow \order[i]$\; \label{line:nextvertex}
$\sol' \leftarrow \sol \cup \set{v}$\; \label{line:sol}
$V' \leftarrow V(H)\cap \neigh{v}$\;
{$\expand(H[V'],\sol')$}\; \label{line:recurse}
$H \leftarrow H\setminus \set{v}$\; \label{line:notv}
}
 \Myif{$\setsize{\sol'} > \setsize{\inc}$}{$\inc \leftarrow \sol'$}\; \label{line:newmax}
\Return\;
}
\end{algorithm2e}
\caption{$\mathtt{Max Clique}(G)$ algorithm}
\label{alg:MaxClique}
\end{algorithm}

The $\mathtt{Max Clique}(G)$ algorithm described in
Figure~\ref{alg:MaxClique},
a generalized version of Algorithm 2.1 in~\cite{McCreesh17Thesis}, 
is a basic maximum clique algorithm which uses a colour-based branch and bound strategy. %
The algorithm keeps the best solution (largest clique) found so far  in the global variable $\inc$ which is initially empty. The current clique is stored in $\sol$ and passed as an argument of the subroutine $\expand$ together with the current subgraph $H\subseteq G$ being considered.
The subroutine $\colourOrder(H)$ (line~\ref{line:colourOrder}) returns 
an ordering of the vertices in $H$, say $v_1, v_2, \ldots, v_n$,  
and for every $i\in[n]$ an upper bound on the number
of colours needed to colour the graph induced by vertices
$v_1$ to $v_i$.

The vertices are then considered in reverse order.
If the vertex $v$ is being considered
and the size of the current growing clique plus the (upper bound on the) number of colours needed to colour the remaining graph is not larger than the current maximum clique (line~\ref{line:bound}) then this branch can be cut. If it is larger,
the algorithm branches on the vertex $v$. First $v$ is taken to be part of the solution: $v$ is added to (a copy of) the current growing solution, (a copy of) the graph is updated to contain only neighbours of $v$ and a recursive call is made (lines~\ref{line:sol}--\ref{line:recurse}). If the recursive call finds a clique larger than the current largest clique, the best solution so far is updated (line~\ref{line:newmax}). The algorithm proceeds to the opposite branch choice, that is, considering vertex $v$ not in the solution (line~\ref{line:notv}). Returning to the loop the algorithm continues to consider the next vertex in the ordering.

It was reported in~\cite{CZ12BranchAndBound} that it is possible to capture the algorithms for solving the maximum clique problem in~\cite{CP90ExactAlgorithm,
	Fahle02SimpleAndFast,
	TS03EfficientBranchAndBound,
	TK07EfficientBranchAndBound,
	KJ07ImprovedBranchAndBound,
	TSHTW10SimpleAndFaster} in a same framework. 
The general algorithm they present is an iterative 
version of the $\mathtt{Max Clique}(G)$ algorithm.
We observe that $\mathtt{Max Clique}(G)$ 
captures also the more recent algorithms in~\cite{ST10NewImplicitBranching,
	SRJ11ExactBitParallel,
	SMRH13ImprovedBitParallel,
	SLB14InitialSorting,
	SLBNP16ImprovedInitialVertex,
	TYHNIW16MuchFaster}. 
The differences in these algorithms reside in the
colouring procedure and in
how the graph operations are implemented
(see~\cite{Prosser12Exact,McCreesh17Thesis} for details). 
For our purpose, that is, in order to show that the running time of
these algorithms can be bounded from below by the length of
the shortest regular resolution refutation of the $k$-clique
formula, we assume
that the colouring algorithm and the graph
operations take constant time and prove
the lower bound for this general framework.
Moreover, 
we can assume that optimal colouring bounds and optimal
ordering of vertices are given.

We now argue that the running time of the $\mathtt{MaxClique}(G)$ algorithm is 
bounded from below by the size of a 
regular resolution refutation of 
$\cliqueblock{G}{k}$
up to a multiplicative factor of $2^kn^{\bigoh{1}}$.
We first note that a straightforward modification
of the $\mathtt{MaxClique}(G)$ algorithm gives
an algorithm, which we refer to as $\mathtt{Clique}(G,k)$,
that determines whether $G$ contains a $k$-clique.
Given a graph $G$ that does not
contain a $k$-clique, an execution of $\mathtt{Clique}(G,k)$
can be represented by a search tree with
leafs labelled by a subgraph $H\subseteq G$ of 
potential clique-members and a number $q$ such that
the branch leading to this leaf has identified 
$k-q$ clique members, has not queried 
any vertex of $H$, and $H$ is $(q-1)$-colourable.
Note that a read-once branching program can simulate this search tree
and, by Proposition~\ref{stm:upperbound_colourable} and the equivalence between read-once branching programs and regular resolution, at each leaf establish that $H$ does not 
contain a $q$-clique in size at most $2^q\cdot q^2 \cdot \setsize{V(H)}^2$. The bound on the running time follows directly.

Observe that establishing that $H$ does not 
contain a $q$-clique is done in a read-once fashion 
by querying only vertices of~$H$. 
Since the vertices of $H$ were not queried earlier on this branch, 
the whole branching program is read-once.
\section{Concluding Remarks}
\label{sec:open-problems}

In this paper we prove optimal average-case lower bounds for regular
resolution proofs certifying $k$-clique-freeness of \ErdosRenyi graphs
not containing $k$-cliques. These lower bounds are also strong enough to
apply for several state-of-the-art clique algorithms used in practice.

The most immediate and 
compelling
question arising from this work is
whether the lower bounds for regular resolution can be strengthened to
hold also for general resolution.  A closer study of our proof reveals
that there are several steps that rely on regularity. 
However, there is no connection per se between regular resolution
and the abstract combinatorial property of graphs that we show to be
sufficient to imply regular resolution lower bounds.  Thus, it is
tempting to speculate that this property, or perhaps some modification
of it, might be sufficient to obtain lower bounds also for general
resolution.
If so, a natural next step would be to try to extend the lower bound further
to the polynomial calculus proof system capturing Gröbner basis calculations.
It is worth mentioning that proving a general resolution lower bound of
$n^{\bigomega{k}}$ for the $k$-clique formula would have 
interesting consequences in parameterized proof 
complexity~\cite{DMS11ParameterizedProofCplx}.

Another intriguing question is whether the lower bounds we obtain asymptotically almost surely for random graphs can also be shown to hold deterministically under the weaker assumption that the graph has certain pseudorandom properties. 
Specifically, is it possible to get an $n^{\bigomega{\log n}}$ \lengthsize lower bound for the class of Ramsey graphs?
A graph on~$n$ vertices is called \emph{Ramsey} if it has no set of
$\lceil{2\log_2 n}\rceil$ vertices forming a clique or an independent set.
It is known that for sufficiently large~$n$ a random graph sampled
from $\randG{\np}{1/2}$ is Ramsey with high probability. Is it true that
for a Ramsey graph~$G$ on~$n$ vertices the formula
$\clique{G}{\lceil{2\log_2 n}\rceil}$ requires (regular) resolution
refutations of \lengthsize $n^{\bigomega{\log n}}$?
The main difficulty towards adapting our argument to this setting is
that Ramsey graphs are, in some sense, less well structured than
random graphs. For example, a random graph plus a constant number of
isolated vertices is, with high probability, still a Ramsey graph, but
it no longer satisfies the first property of clique-denseness
(Definition~\ref{def:propertyP}). This particular problem can be
circumvented using a result from\ \cite[Theorem 1]{PR.99}---as was
done in~\cite{LPRT17ComplexityRamsey} to obtain a lower bound for
tree-like resolution---but proving that a Ramsey graph satisfies the
second property of clique-denseness, or some suitable version of it, 
seems significantly
more challenging.

\paragraph{Acknowledgements}
This work has been a long journey, and 
different subsets of the authors want to acknowledge fruitful and
enlightening discussions with different subsets from the following list
of colleagues:
Christoph~Berkholz,
Olaf~Beyersdorff,
Nicola~Galesi,
Ciaran~McCreesh,
Toni~Pitassi,
Pavel~Pudl\'{a}k,
Ben~Rossman,
Navid~Talebanfard, 
and
Neil~Thapen.
A special thanks to Shuo Pang for having pointed out an inaccuracy in the probabilistic argument 
in Section~\ref{sec:lowerbound2} and having suggested a fix.

The first, second, and fourth authors were
supported by the European Research Council
under the European Union's \mbox{Horizon} 2020 Research and Innovation
Programme / ERC grant agreement no.~648276 AUTAR.

The third  and fifth authors were supported
by the European Research Council under the
European Union's Seventh Framework Programme \mbox{(FP7/2007--2013) /}
ERC grant agreement no.~279611 as well as by 
Swedish \mbox{Research} Council grants
\mbox{621-2012-5645}
and
\mbox{2016-00782},
and the second author did part of this work while at 
KTH Royal Institute of Technology supported by the same grants.
The last author was supported by the Russian Foundation for Basic Research.
%


\newcommand{\etalchar}[1]{$^{#1}$}

\end{document}